\documentclass[a4paper,twocolumn,superscriptaddress,11pt,accepted=2019-06-25]{quantumarticle}
\pdfoutput=1
\usepackage[utf8]{inputenc}
\usepackage[english]{babel}
\usepackage[T1]{fontenc}
\usepackage{amsmath}
\usepackage{hyperref}
\usepackage{amssymb}
\usepackage{amsthm}
\usepackage{bm}
\usepackage{braket}
\usepackage{adjustbox}
\usepackage{array}
\usepackage{qcircuit}
\usepackage{subcaption}
\usepackage[numbers,sort&compress]{natbib}

\PassOptionsToPackage{svgnames}{xcolor}

\usepackage{tikz}
\usetikzlibrary{shapes,snakes}
\usetikzlibrary{calc,3d, arrows.meta, decorations.markings,math}

\usepackage{lipsum}

\newtheorem{theorem}{Theorem}
\newtheorem{definition}{Definition}
\newtheorem{conjecture}{Conjecture}
\newtheorem{lemma}{Lemma}
\newtheorem{corollary}{Corollary}

\begin{document}

\title{Nonadaptive fault-tolerant verification of quantum supremacy with noise}
\date{\today}
\author{Theodoros Kapourniotis}
\email{T.Kapourniotis@warwick.ac.uk}
\orcid{0000-0002-6885-5916}
\affiliation{Department of Physics, University of Warwick, Coventry CV4 7AL, United Kingdom}
\author{Animesh Datta}
\email{animesh.datta@warwick.ac.uk}
\orcid{0000-0003-4021-4655}
\affiliation{Department of Physics, University of Warwick, Coventry CV4 7AL, United Kingdom}

\maketitle

\begin{abstract}
  Quantum samplers are believed capable of sampling efficiently from distributions that are classically hard to sample from.
We consider a sampler inspired by the classical Ising model. It is nonadaptive and therefore experimentally amenable. Under a plausible conjecture, classical sampling upto additive errors from this model is known to be hard. We present a trap-based verification scheme for quantum supremacy that only requires the verifier to prepare single-qubit states. The verification is done on the same model as the original sampler, a square lattice, with only a constant overhead. We next revamp our verification scheme in two distinct ways using fault tolerance that preserves the nonadaptivity.  The first has a lower overhead based on error correction with the same threshold as universal quantum computation. The second has a higher overhead but an improved threshold (1.97\%) based on error detection. We show that classically sampling upto additive errors is likely hard in both these schemes. Our results are applicable to other sampling problems such as the Instantaneous Quantum Polynomial-time (IQP) computation model. They should also assist near-term attempts at experimentally demonstrating quantum supremacy and guide long-term ones.
\end{abstract}

\section{Introduction}
\label{sec:intro}

Considerable experimental efforts are being directed towards the realisation of quantum information processing technologies, with the eventual aim of constructing  universal quantum computers and simulators~\cite{RevModPhys.86.153, fukuhara2013quantum, lanting2014entanglement, kelly2015state}. One of the motivations for this exercise is their expected ability to simulate physical systems that are believed to be intractable classically~\cite{feynman1982simulating, buluta2009quantum}.  An example of such a system is a lattice of interacting spins in the presence of a magnetic field, represented by the classical Ising model \cite{ising1925beitrag}, which is a workhorse in condensed matter physics and statistical physics \cite{binder1986spin,savary2016quantum, lee1952statistical}.

\emph{Computing} the partition function of the Ising model in an external magnetic field is, however, \#P-hard even in two spatial dimensions with multiplicative approximations~\cite{goldberg2014complexity,fujii2013quantum}. \#P-hard problems are not expected to be solvable efficiently on a universal quantum computer. 
\textit{Sampling} up to multiplicative~\cite{8467641,bremner2010classical,farhi2017quantum, boixo2018characterizing} and additive~\cite{d71a5fed54df414db082ed64ed5c9ef7, aaronson10boson_sampling,gao_quantum_2016, miller2017quantum,morimae2017hardness,bermejo2017architectures} errors from certain distributions, such as from the partition function of the Ising model at imaginary temperatures, is possible using devices that do not require the full complement of DiVincenzo's criteria~\cite{divincenzo2000physical} for their implementation. Their scalable implementation is thus anticipated to be more achievable than a universal quantum computer's, providing tangible demonstrations of quantum supremacy sooner~\cite{spring2013boson, latmiral2016towards,boixo2018characterizing, bremner2016achieving, fujii2016noise}. This expectation is purchased at the price of sacrificing the full power of universal quantum computers, promising only to efficiently sample from certain distributions instead.
	This is the remit of `quantum supremacy'~\cite{preskill2012quantum} but the relative experimental ease introduces new theoretical challenges~\cite{Harrow2017}.

Given the significance of quantum supremacy in wider quantum information science, demonstrating it experimentally is vital. In the real world however, this faces two crucial challenges~\cite{AaronsonChen2016,Harrow2017} that arise from experimental imperfections and noise respectively.
The first is verifying that the output distribution is correct, or at least close to correct.
Since any real world experiment will be imperfect, establishing correctness is important, particularly so since sampling, unlike for instance integer factoring, is not in NP whose correctness can be checked efficiently.
This calls for a verification scheme with minimal trust assumptions, that indicates whether or not the output distribution is sufficiently close in total variation distance to the ideal one, and consequently whether or not an experiment has successfully demonstrated quantum supremacy. 
The hurdle is to ensure that the verifiable supremacy experiment no more demanding than the original non-verifiable experiment. 
This requires a redesign of verification schemes for universal quantum computing~\cite{fitzsimons12vubqc,
aharonov10abe_published,
broadbent15verification}.

The second challenge arises because all experiments are noisy.
Arguing for experimental quantum supremacy is incumbent on the sampling task being classically hard to simulate even in the presence of noise. 
As excessive noise can render a hard probability distribution easy to simulate, it is an important challenge to determine to what extent a sampling task remain hard to simulate classically, even in the presence of noise~\cite{Harrow2017}. The hurdle therefore is to retain hardness in the presence of noise. This is what fault-tolerance provides. Of course, this is only worthwhile for experiments if the fault-tolerance threshold for verifiable quantum supremacy is strictly easier to achieve than that of universal quantum computation. Furthermore, the overheads needed for fault-tolerant verifiable quantum supremacy must be less demanding than that for universal quantum computation. Fault-tolerant quantum supremacy is thus a compelling milestone on the way to a fault-tolerant universal quantum computer.

Among the numerous quantum supremacy candidates~\cite{Harrow2017} are IQP~\cite{bremner2010classical} and those allowing sampling from the distribution of partition functions of the classical Ising model at imaginary temperature - the Ising sampler~\cite{gao_quantum_2016}.
IQP and the Ising sampler exhibit in line with other models~\cite{ aaronson10boson_sampling, fujii2013quantum, farhi2017quantum, boixo2018characterizing}, under plausible conjectures, a highly unexpected collapse of the polynomial hierarchy to the third level occurs if a classical sampler can sample the partition function distribution upto additive errors. 
The Ising Sampler is a constant-depth version of IQP, with the additional favourable properties of translational-invariance and single-instance.  Single-instance means that the hardness results hold for a single fixed instance of the problem, relieving the burden of creating random instances such as Boson Sampling~\cite{aaronson10boson_sampling} and Random Sampling~\cite{boixo2018characterizing}. This simplifies both theoretical analysis and experimental implementation.
Although the partition function at imaginary temperatures may appear unphysical, it has deep connections to quantum complexity theory~\cite{goldberg2014complexity} as well as quantum statistical and condensed matter physics via analytic continuations~\cite{lee1952statistical}.
It is the combination of relative ease in theoretical analysis and experimental implementations allied with its strong connections to physics that motivates our choice of the Ising Sampler as the subject of this paper.

In this paper, we provide a fault-tolerant scheme for the verification of quantum supremacy in the Ising sampler. We achieve this by  amalgamating trap-based quantum verification  techniques~\cite{fitzsimons12vubqc} with recent results on demonstrating quantum supremacy by emulating fault tolerance via post-selection~\cite{fujii2016noise}. 
In response to the first challenge above, we present a nonadaptive verification scheme with exponentially low probability of failure and only linear complexity for the Ising sampler (prover). 
Our scheme applies to any untrusted prover with entangling and measuring capabilities, limited only by the laws of quantum mechanics, and requires the verifier to prepare random, single-qubit states with bounded local noise.
We first present a verification scheme that should aid demonstrating quantum supremacy with few qubits (Theorem~\ref{thm1}).
In response to the second challenge above, we prove fault-tolerant versions (Theorem~\ref{thm2}), one of which uses the idea of emulated fault tolerance by post-selection~\cite{fujii2016noise}. This `free' post-selection enables us to provide a fault-tolerant verification scheme with improved thresholds over universal quantum computing thresholds. An important property of our verification is that it itself is within the instantaneous model of quatum computing and therefore can be implemented in the same device as the Ising sampler with small modifications. Moreover, we prove quantum supremacy of this modified model (Theorem~\ref{thm3}).

\subsection{Comparison to prior work and structure}

We go beyond Ref.~\cite{fitzsimons12vubqc} in three ways, namely 
	(i) providing a new definition of verifiability over many i.i.d. repetitions of the protocol, based on the total variation distance between the output distribution and the correct one; 
	(ii) using the Raussendorf-Harrington-Goyal (RHG) strategy for fault tolerance instead of using it for probability amplification; and 
	(iii) developing a simpler construction for verification and computation on separate planar graphs.

 We go beyond Ref.~\cite{fujii2016noise} by (i) solving one of its open problems -- proving quantum supremacy of improved threshold fault-tolerant model up to \emph{additive} errors (Theorem~\ref{thm3}) and (ii) verifying quantum supremacy while maintaining its improved thresholds.

Our work is a significant improvement over the verification method of Ref.~\cite{gao_quantum_2016} on two counts. Firstly, our schemes require a linear overhead in the number of qubits as opposed to a quadratic one in Ref~\cite{gao_quantum_2016}. Secondly, our schemes (fault-tolerant ones) scale in the face of constant local noise while that of Ref~\cite{gao_quantum_2016} requires local noise polynomially small in the total number of qubits.  Our verification schemes apply to any nonadaptive sampler based on cluster states, however we will use the Ising sampler~\cite{gao_quantum_2016} as a particular example, thus keeping the benefit of the single instance property and experimental feasibility of this model.

Our work is structured as follows. Section~\ref{section2} defines verifiability based on the total variation distance of the output distribution. Section~\ref{sec:sampver} provides an overview of the Ising sampler placed in an a cryptographic setting of a prover and a verifier. Section~\ref{main:nft}   contains the first of our main results on the verification of the Ising sampler's output. We present a non-fault-tolerant verification scheme (Theorem~\ref{thm1}). Since it requires decreasing noise in preparation, entanglement and measurement with increasing system size, this is only viable in small-sized experimental demonstrations of quantum supremacy. In Section~\ref{sec:FTver} we present two fault-tolerant versions of the verification scheme (Theorem~\ref{thm2}) which are scalable when noise is below certain thresholds. Section~\ref{sec:supr} provides a result (Theorem~\ref{thm3}) on the quantum supremacy of the output distribution  of the noisy Ising sampler conditioned on syndrome measurements accepting, which is a generalisation of Ref.~\cite{fujii2016noise} for additive errors.

\section{Verifiability}
\label{section2}

We begin with our definition of a verification protocol.

\begin{definition}[Verification protocol]
\label{def_prot}
A verification protocol involves two parties - a trusted verifier and an untrusted prover who share a quantum and classical channel. 
The protocol takes as input a description of a computation and outputs a string and a bit. 
The bit determines if the string is accepted or rejected. 
\end{definition}

Establishing verifiability of a protocol consists of proving completeness and soundness.
A protocol is complete if, for an honest prover, the verifier outputs the correct result and accepts. 
A protocol is sound if, for any deviation of the prover, the probability that the verifier outputs an incorrect result and accepts is low. This deviation captures both a malevolent prover who tries to cheat and uncontrollable errors in the prover's device. 

Note that the above notion of verifiability relies on an output string being correct while sampling relies on distributions being close. We are therefore interested in the total variation distance between the experimental output distribution and the exact one~\cite{fitzsimons12vubqc}. We are furthermore interested in arguing for quantum supremacy based on the total variation distance between distributions. This requires us to go from a joint distribution of a string and a bit to a probability distribution over strings conditioned on a bit. To meet these demands we introduce the idea of a verification scheme, that uses a protocol as a black box and can call it repeatedly. We also assume that the repetitions of the protocols are independent and identically distributed (i.i.d.). However, there is no assumption on the behaviour of the system within the protocol, which means that an adversarial prover can cheat by correlating systems within the protocol.

\begin{definition}[Verification scheme]
\label{def_sch}
A verification scheme takes as input a verification protocol, $M \in \mathbb{N}, l \in [0,1]$  and outputs a string and a bit. The bit determines if the string is accepted or rejected.

A verification scheme works as follows. After running $M$ i.i.d repetitions of a verification protocol it outputs one of the $M$ output strings at random and accepts if at least a fraction $l$ of the protocols accept and rejects otherwise.
\end{definition}

Let  $q^{\text{nsy}}(\bm{x})$ be the experimental and $q^{\text{exc}}(\bm{x})$ be the exact distribution   of the output $\bm{x}$ of a sampler. We are interested in the quantity
\begin{equation}
\label{eq:var_def}
\text{var}\equiv \frac{1}{2} \sum_{\bm{x}} | q^{\text{exc}}(\bm{x}) - q^{\text{nsy}}(\bm{x}) |,
\end{equation} 
  where the sum is over all binary strings $\bm{x}$ of size $N$.

The following definition captures the notions of completeness and soundness at the level of a scheme for sampling problems.

\begin{definition}[Verifiability of a scheme]
\label{def_ver}
A scheme is verifiable if its output is

\begin{itemize}
\item $(\delta',\delta)-$complete: For an honest prover having only bounded noise, the scheme accepts at least with probability $\delta',$ and 
\begin{equation}
\mathrm{var} \leq 1 - \delta
\end{equation}
for the output string.

\item $(\varepsilon',\varepsilon)-$sound: For any, including adversarial, prover if the scheme accepts then
\begin{equation}
\mathrm{var} \leq \varepsilon
\end{equation}
with confidence $\varepsilon'.$

\end{itemize}
\end{definition}

We then consider the verifiability of a scheme for a sampler which has a designated output register we call the post-selection register. We consider probabilities $q^{\text{nsy}}(\bm{x}|y=0)$, where $y$ is the value of the post-selection register, for the experimental and $q^{\text{exc}}(\bm{x}|y=0)$ for the exact distribution of a sampler, conditioned on $y$ being zero. We are interested in the quantity
\begin{equation}
\text{var}^{\text{Post}}\equiv \frac{1}{2} \sum_{\bm{x}} | q^{\text{exc}}(\bm{x}|y=0) - q^{\text{nsy}}(\bm{x}|y=0) |.
\end{equation}
Again, the sums are over all binary strings $\bm{x}$ of size $N$. We adapt our definition to conditional probabilities as follows.

\begin{definition}[Verifiability of a scheme for post-selected distribution]
\label{def_ver_post}
A scheme is verifiable conditioned on the post-selection register being zero, if its output is

\begin{itemize}
\item $(\delta',\delta)-$complete: For an honest prover having only bounded noise, the scheme accepts at least with probability $\delta',$ and 
\begin{equation}
\mathrm{var^{Post}} \leq 1 - \delta
\end{equation}
for the the output string.

\item $(\varepsilon',\varepsilon)-$sound: For any, including adversarial, prover if the scheme accepts, then
\begin{equation}
\mathrm{var^{Post}} \leq \varepsilon
\end{equation}
with confidence $\varepsilon'.$

\end{itemize}
\end{definition}

\section{Quantum sampling in the verifier-prover setting}
\label{sec:sampver}

In the verifier-prover setting, the verifier can prepare bounded-error, single-qubit states and the prover implements the rest of the computation including the measurements, and returns the output samples to the verifier. The role of the verifier is to ascertain if the prover is acting honestly and executing the correct operation. The prover is, in general, malicious, trying to pass any tests designed by the verifier while deviating from the correct implementation at the same time. This malice may be intentional if the prover is trying to convince the verifier of its quantum power when it has none, or incidental if the prover possesses an imperfect quantum device prone to noise and errors. We assume that the prover's deviations are governed by quantum mechanics.

We begin by adapting the Ising spin model to a blind verifier-prover cryptographic setting~\cite{broadbent09ubqc}. Blindness, which ensures that the prover remains ignorant of the actual computation, is a necessary ingredient in our verification scheme.
Our Ising spin model consists of qubits in state $\ket{+}$ subject to nearest neighbour controlled $\pi$-phase rotations, denoted by $cZ$. All the qubits are measured simultaneously in a basis in the $xy$-plane of the Bloch sphere. The measurement outcome of classical bits is the output sample. This model corresponds to the well-studied measurement-based quantum computing (MBQC) model~\cite{raussendorf00mbqc} without the adaptive measurements. This last restriction makes the depth of the computation constant on the size of the input: one round of preparation, three rounds of entangling because the maximum degree of the graph is three and one round of non-adaptive measurements. This relaxes DiVincenzo's criteria of long decoherence times and makes this model non-universal for quantum computing.

In the particular Ising sampler presented in \cite{gao_quantum_2016} the structure of the graph state is fixed (Fig.~(\ref{brick1})), but its size scales with the width $m$ and depth $n.$ The measurement angles are also fixed to specific values from the set $\{-\pi/4,  -\pi/8, 0,\pi/8, \pi/4\}$. This choice of graph, which we call the `extended' brickwork state, and a fixed angle for each physical qubit has the following benefit: Each possible combination of measurement outcomes `chooses' a different angle for each qubit of the original brickwork state from the set $\{k \pi/4\}$, $k = \{0, \ldots, 7\}$. This effectively makes a single instance of the model a random quantum circuit generator, a property exploited to prove its hardness.

The correspondence to an Ising model comes from the locality of spin interactions and decomposing each MBQC measurement into a unitary rotation around the $z$-axis corresponding to an external magnetic field, followed by a Pauli $X$ measurement. The quantum state just before the Pauli $X$ measurement is given by the unitary evolution due to the Hamiltonian
\begin{equation}
\mathcal{H} = - \sum_{\langle i,j \rangle} J Z_i Z_j + \sum_i B_i Z_i 
\end{equation}
where $J$ is the interaction term, $B_i$ the local field strength and $Z_i$ the Pauli $Z$ operator on qubit $i$. 

The probability $q^{\text{exc}}(\boldsymbol{x})$ of measuring a bit string $\boldsymbol{x}$  corresponds to the partition function $\mathcal{Z}_{\boldsymbol{x}}$ of the Ising model with Hamiltonian $\mathcal{H}' \equiv \mathcal{H} + \frac{\pi}{2} \sum_i x_i Z_i$ and is given by 
\begin{equation}
q^{\text{exc}}(\boldsymbol{x}) =\frac{|\text{Tr} (e^{-i (\mathcal{H} + \frac{\pi}{2}\sum_i x_i Z_i)} )|^2}{2^{2N}} \equiv \frac{|\mathcal{Z}_{\boldsymbol{x}}|^2}{2^{2N}},
\end{equation}
where $N=mn$. The second term in $\mathcal{H}'$ comes from the measurement outcomes of the Pauli $X$ measurements, and the partition function is evaluated at an imaginary temperature $\beta = 1 / k_B T = i$.  

Testing the honesty of the prover, in our case the Ising sampler, requires the `blind' injection of certain `trap' qubits. To keep the identity of these trap qubits from the prover, the verifier applies some encoding on the original translationally-invariant Ising spin model, making the model translationally variant. Now both the participating qubits and the measurement angles on the graph state have a randomly chosen extra rotation according to the scheme described next.

Specifically, each qubit $i$ is individually prepared by the verifier in the state  $\ket{+_{\theta_i}},$ where $\theta_i$ is chosen uniformly at random from the set $A=\{0, \frac{\pi}{8},\frac{2\pi}{8}, \ldots,  \frac{15\pi}{8}\}.$ Instead of the prover measuring in fixed predetermined angles, as in the original Ising sampler, the verifier sends  encrypted angles to the prover: $\delta_i = \theta_i + (-1)^{r'_i}\phi_i + r_i \pi$ for $r_i, r'_i \in_R \{0,1\}$, where $\in_R$ stands for a uniform random selection. Rotations by $\theta_i$ on the qubit and on the angles mutually cancel and the classical information that the prover receives (containing the actual measurement angles $\phi_i$) is classically one time padded by $\theta_i$. The bits $r_i,r'_i$ provide some extra randomness to restrict the information the prover gets from the quantum state and can be corrected by classical post-processing of the sample. Our difference from Ref.~\cite{broadbent09ubqc} lies in the number of angles used in the set $A$, and comes from the fact that we use a different decomposition of the computation. We conjecture that this can be improved upon (See Sec. \ref{sec:Disc}).

\section{Non-fault-tolerant Verification of Ising Sampler} \label{main:nft}

The output of a quantum sampler must be classical for it to be comparable to that of a classical sampler, a prerequisite for demonstrating quantum supremacy. This allows us to simplify trap-based verification strategies for universal quantum computation~\cite{fitzsimons12vubqc,kapourniotis15linear_verification,kashefi2015optimised} to having disjointed computational resource and trap states - an idea also used in Ref.~\cite{kashefi2015optimised} and in circuit-based verification~\cite{broadbent15verification}.
This permits an exponentially small error in our estimation of the fidelity of the output using a square lattice. 
Finally, a trap-based technique instead of fidelity-witness based certification ones~\cite{hangleiter2016direct,gao_quantum_2016}, similar also to \cite{aolita2015reliable,cramer2010efficient}, enables us to reduce the resource complexity of the verification protocol from quadratic to linear. Other certification methods that require linear resources exist, by trusting the measuring devices~\cite{hayashi15,markham2018simple,
pallister2018optimal,ferracin18} instead of preparation. Linear resource complexity is minimal in this scenario because the verifier needs to receive at least one copy of the resource state to perform the computation.

\begin{figure}[t]
\centering
\resizebox{\linewidth}{!}{
\begin{tikzpicture}

\pgftransformcm{1}{0}{0.4}{0.5}{\pgfpoint{0cm}{0cm}};
\foreach \x in {10,12,14,16,18,20,22,24,26} {
	\foreach \y in {1,3,5,7} {            

	\node at (\x,\y) [circle,draw=black,thick,fill=white] {};

	}
}

\draw [very thick] (10,5) -- (26,5);
\draw [very thick] (14,5) -- (14,7);
\draw [very thick] (10,7) -- (26,7);
\draw [very thick] (18,5) -- (18,7);

\draw [very thick] (10,1) -- (26,1);
\draw [very thick] (14,1) -- (14,3);
\draw [very thick] (10,3) -- (26,3);
\draw [very thick] (18,1) -- (18,3);

\draw [very thick] (22,3) -- (22,5);
\draw [very thick] (26,3) -- (26,5);

\end{tikzpicture}
}

\vspace{2ex}

$(i)$

\vspace{5ex}

\resizebox{\linewidth}{!}{
\begin{tikzpicture}

\pgftransformcm{1}{0}{0.4}{0.5}{\pgfpoint{0cm}{0cm}};

\node at (12,3) {\Huge =};

\node at (12,1) {\Huge =};

\foreach \x in {14,16,18,20,22,24,26} {
	\foreach \y in {1,3} {            

	\node at (\x,\y) [circle,draw=black,thick,fill=black] {};

	}
}

\node at (10,3) [circle,draw=black,thick,fill=white] {};

\node at (10,1) [circle,draw=black,thick,fill=white] {};

\draw [very thick] (13,3) -- (27,3);

\draw [very thick] (13,1) -- (27,1);

\draw [very thick] (10,1) -- (10,2);
\draw [very thick] (14,1) -- (14,2);

\draw [very thick] (9,1) -- (11,1);
\draw [very thick] (9,3) -- (11,3);

\node at (14,4) { \huge $\pi/8$};
\node at (16,4) { \huge $0$};
\node at (18,4) { \huge $-\pi/4$};
\node at (20,4) { \huge $0$};
\node at (22,4) { \huge $\pi/4$};
\node at (24,4) { \huge $0$};
\node at (26,4) { \huge $-\pi/8$};

\end{tikzpicture}

}

\vspace{2ex}

$(ii)$

\caption{The original brickwork state $(i)$ is a universal resource for MBQC under $xy$-plane measurements, where white vertices represent qubits and the edges represent $cZ$ operations. The `extended' brickwork state $(ii)$ is used in the original Ising sampler~\cite{gao_quantum_2016}, where each white vertex is replaced by $7$ physical qubits (black vertices). The measurement angle for each qubit is fixed to the value written above each vertex. There is no adaptation of the angles based on previous measurement outcomes as in universal MBQC.}
\label{brick1}
\end{figure}
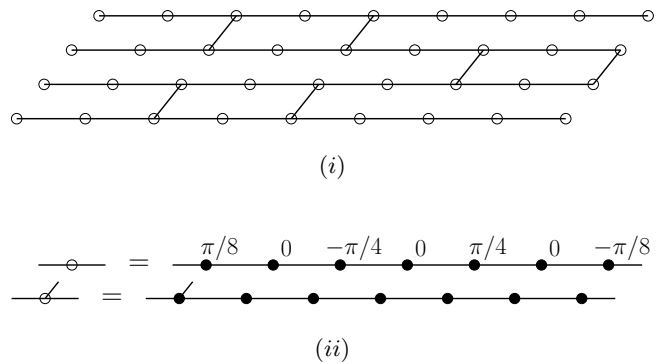

Our verification protocol relies on judiciously selecting the measurement angles and placing \emph{dummy} qubits prepared in the state $\ket{0}$, which, together with $xy$-plane measurements, allows us to \emph{carve} different types of graphs from a square lattice graph, as shown in Fig.~(\ref{protocol_pic}). Placing one dummy qubit between any two other qubits prevents the prover's entangling operators to have any entangling effect between the participating qubits, so the prover can apply exactly the same operations that produce the square lattice but create a different graph state. The graphs carved out are  the `extended' brickwork state (Fig.~(\ref{brick1})) and two other graphs containing special `trap' qubits in the state $\ket{+}.$ The extended brickwork state is used to run the Ising sampler. The traps in the trap graphs are measured in the same basis as prepared, yielding a deterministic check on the prover. Two different types of the trap graphs are needed to enable placing a trap at any position in the graph with equal probability.
The order of the graphs is chosen at random and the whole protocol implemented blindly to thwart the prover from distinguishing trap and target computation qubits.

A sketch of our Protocol 1 appears in Fig.~(\ref{protocolss}), and the details in Sec.~\ref{NFTverif}.
The protocol has constant time complexity of the quantum operations and needs $O(N)$ qubits, where $N$ is the number of the qubits of the Ising sampler.

Noise considered in all our protocols for an honest prover is local, unital and bounded. It applies after every elementary operation (preparation, entangling and measurement) $j$ and is expressed as a CPTP superoperator:

\begin{equation}
\mathcal{N}_j = (1-\epsilon_{V,P}
) \mathcal{I} + \mathcal{E}_j
\label{noise_eq}
\end{equation}
where $\epsilon_{V,P}= \epsilon_{V}$ and $||\mathcal{E}_j||_{\diamond}=\epsilon_V$ for the noise of the verifier (preparation noise) and  $\epsilon_{V,P}= \epsilon_{P}$ and $||\mathcal{E}_j||_{\diamond}=\epsilon_P$ for the noise of the honest prover (entangling and measurement noise).

\begin{figure}[t]
\centering
\begin{tikzpicture}
	\node [draw,thick,align=left,text width = 7.5cm] at (0,0) {
		1. Verifier selects a random ordering of \( 2 \kappa + 1 \) graphs, one for target computation and \( \kappa \) from each type of trap graphs.\\
		\strut\\
		2. Verifier prepares, one by one, the qubits needed for the blind implementation of the \( 2 \kappa + 1 \) cluster states and sends them to the prover.\\
		\strut\\
		3. Verifier sends the encrypted measurement angles to the prover.\\
		\strut\\
		4. Prover entangles all received qubits in the \( 2 \kappa + 1 \) cluster states.\\
		\strut\\
		5. Prover measures all qubits simultaneously in the instructed angles and returns the results.\\
		\strut\\
		6. Verifier decrypts the outputs and accepts if all trap results are correct, otherwise rejects.\\
	};
	\node at ($(current bounding box.south west)+(-0.1,-0.1)$) {};
	\node at ($(current bounding box.north east)+(+0.1,+0.1)$) {};
\end{tikzpicture}

\caption{Nonadaptive verification protocol}
\label{protocolss}
\end{figure}

\vspace{0.7cm}
\begin{theorem}[Non-fault tolerance verification scheme]
\label{thm1}
There exists a verification scheme with Protocol 1, $M = \log(1/\beta)/(2 \kappa^2 N^2 (\epsilon_V + \epsilon_P)^2)$, $l=(1-\kappa N (2 \epsilon_V + 4  \epsilon_P))$ that according to Def.~(\ref{def_ver}) is
 $$\left( 1-\beta,1-\sqrt{N (\epsilon_V + 3 \epsilon_P)}\right)-complete$$ and 
 $$\left(1-\beta,\sqrt{\kappa N (3 \epsilon_V + 5 \epsilon_P)  + \Delta_{\kappa}}\right)-sound,$$
 where $\Delta_{\kappa}= \kappa! (\kappa+1)!/(2\kappa+1)!$.
\end{theorem}

In the above, $\epsilon_V$ and $\epsilon_P$ are fixed by the experimental capability, while completeness and soundness parameters are set by the conjectures invoked to argue for quantum supremacy, as obtained in Eqn.~(\ref{eq:epsi}).

 A proof sketch appears in Section \ref{sec:sound} and a full proof in Appendix \ref{app:Thm1}.

Using our verifiable quantum sampler to demonstrate quantum supremacy is underwritten by results which show that approximating the Ising sampler upto constant total variation distance is hard classically, subject to an average case hardness and an anti-concentration conjecture, presented in Section~\ref{sec:supr}, similarly to the original model~\cite{gao_quantum_2016}.

Both $\kappa N \epsilon_V$ and $\kappa N \epsilon_P$ must be constant for the total variation distances to be constant, plus exponentially decaying in $\kappa$ term $\Delta_{\kappa}$ in soundness,  in Theorem~\ref{thm1}. To achieve this we require local errors $\epsilon_V$ and $\epsilon_P$ to decrease linearly with the number of qubits and $\kappa$. This is only realistic in quantum supremacy experiments involving a few qubits. 

To overcome this restriction, we consider  fault-tolerant versions of our verification protocol in Section~\ref{sec:FTver}.

\subsection{Protocol}
\label{NFTverif}

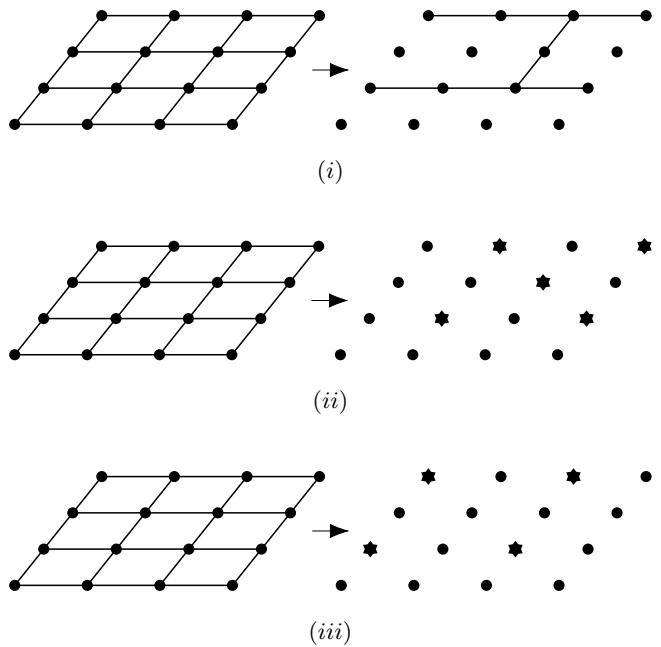
\begin{figure}[t]
\centering

\resizebox{\linewidth}{!}{
\begin{tikzpicture}
\pgftransformcm{1}{0}{0.4}{0.5}{\pgfpoint{0cm}{0cm}};
\foreach \x in {1,3,5,7} {
	\draw[black,very thick] (1,\x) -- (7,\x);
		\draw[black,very thick] (\x,1) -- (\x,7);
\foreach \y in {1,3,5,7} {             \node at (\x,\y) [circle,fill=black] {};
            }}
\draw[-{Latex[scale=3.0]}] (8,4) -- (9,4);

\foreach \x in {10,12,14,16} {
	\foreach \y in {1,3,5,7} {            

	\node at (\x,\y) [circle,fill=black] {};

	}
}

\draw [very thick] (10,3) -- (16,3);
\draw [very thick] (14,3) -- (14,7);
\draw [very thick] (10,7) -- (16,7);

\end{tikzpicture}
}

\vspace{2ex}

$(i)$

\vspace{5ex}

\resizebox{\linewidth}{!}{
\begin{tikzpicture}
\pgftransformcm{1}{0}{0.4}{0.5}{\pgfpoint{0cm}{0cm}};
\foreach \x in {1,3,5,7} {
	\draw[black,very thick] (1,\x) -- (7,\x);
		\draw[black,very thick] (\x,1) -- (\x,7);
            \foreach \y in {1,3,5,7} {             \node at (\x,\y) [circle,fill=black] {};
            }}
   \draw[-{Latex[scale=3.0]}] (8,4) -- (9,4);

\foreach \x in {10,12,14,16} {
            \foreach \y in {3,7} {            

\pgfmathparse{Mod( (\y - 1)/2 + (\x - 10)/2 ,2)==1?1:0}

\ifnum \pgfmathresult   = 1 

\node at (\x,\y) [circle,fill=black] {};

\else \node at (\x,\y) [star,star points=6,fill=black]{};

\fi

	}
}

\foreach \x in {10,12,14,16} {

\node at (\x,1) [circle,fill=black] {};

}

\foreach \x in {10,12,16} {
\node at (\x,5) [circle,fill=black] {};
}

\node at (14,5) [star,star points=6,fill=black] {};

\end{tikzpicture}
}

\vspace{2ex}

$(ii)$

\vspace{5ex}

\resizebox{\linewidth}{!}{
\begin{tikzpicture}
\pgftransformcm{1}{0}{0.4}{0.5}{\pgfpoint{0cm}{0cm}};
\foreach \x in {1,3,5,7} {
	\draw[black,very thick] (1,\x) -- (7,\x);
		\draw[black,very thick] (\x,1) -- (\x,7);
            \foreach \y in {1,3,5,7} {             \node at (\x,\y) [circle,fill=black] {};
            }}
   \draw[-{Latex[scale=3.0]}] (8,4) -- (9,4);
\foreach \x in {10,12,14,16} {
            \foreach \y in {3,7} {            

\pgfmathparse{Mod( (\y - 1)/2 + (\x - 10)/2 ,2)==1?1:0}

\ifnum \pgfmathresult   = 0 

\node at (\x,\y) [circle,fill=black] {};

\else \node at (\x,\y) [star,star points=6,fill=black]{};

\fi

	}
}

\foreach \x in {10,12,14,16} {

\node at (\x,1) [circle,fill=black] {};
\node at (\x,5) [circle,fill=black] {};

}

\end{tikzpicture}
}

\vspace{2ex}

$(iii)$

\caption{Verifier chooses a random ordering of $2\kappa+1$ graph states - the computational graph on the right of figure $(i)$; $\kappa$ identical trap graphs on the right of figure $(ii)$ which have traps (starred nodes) on the even parity positions of the sub-graph that corresponds to the computational graph; and $\kappa$ identical graphs on the right of figure $(iii)$ which have traps on the odd parity positions of the sub-graph that corresponds to the computational graph. All of these graphs can be generated from a square lattice (on the left) by replacing $\ket{+}$ qubits with $\ket{0}$ at the positions (isolated dot nodes) we do not want {entangled} with their neighbours when $cZ$ is applied. Further detail on the carving procedure, which can be made blind (Section \ref{sec:sampver}), is provided in Appendix \ref{app:bridge}. }
\label{protocol_pic}
\end{figure}

The following is a full description of the non-fault tolerant verification protocol:

\vspace{0.3cm}

\textbf{Protocol 1:}

\begin{enumerate}
\item Verifier selects a random ordering of $2 \kappa+1$ graphs, one for computation and $2\kappa$ for testing, as in Fig.~(\ref{protocol_pic}).
This fixes the position of computational basis qubits called the \emph{dummy} qubits (see Appendix \ref{app:bridge}) and the measurement angles $\{\phi_i\}_{i=1}^N$, where $N=m  \times  n$ is the total number of qubits, so that

(a) in the target computation graph we carve from the square lattice a universal resource state, the `extended' brickwork state of Figure~(\ref{brick1}) and fix the rest of the measurement angles according to the Ising sampler model;

(b) in the trap computation graphs the dummy qubits are used to isolate the traps, which are placed at fixed positions. For half of the graphs in positions with odd parity that correspond to non-dummy qubits in the computational graph and in the other half in positions with even parity that correspond to non-dummies in the computational graph. The traps are measured with angles $\phi=0$ so that the measurement is deterministic. Crucially, the trap graphs do not contain any `bridge' operations so there is no need for adaptive corrections.
 
\item Verifier prepares the qubits that compose the cluster state one by one and sends them to the prover.

(a) The dummy qubits are prepared in $\{\ket{d_i}:d_i \in_R \{0,1\}\}.$ 

(b) The rest of the qubits are prepared in $Z^{d_{k \sim j}}\ket{+_{\theta_j}}= \frac{1}{\sqrt{2}} ( \ket{0} + e^{i (\theta_j + d_{k \sim j} \pi)} \ket{1})$, where $\theta_j$ is chosen uniformly at random from the set $A=\{0, \frac{\pi}{8},\frac{2\pi}{8}, \ldots,  \frac{15\pi}{8}\}\}$ and $d_{k \sim j}$ is the parity of the $d$'s of all neighbours of $j$. Notice that the set $A$ is different from the original trap-based protocol of Ref. \cite{fitzsimons12vubqc}.
\item Verifier sends the encrypted computational measurement angles to the prover: $\delta_i = \theta_i + (-1)^{r'_i}\phi_i + r_i \pi$ for $r_i, r'_i \in_R \{0,1\}$. Parameters $r_i, r'_i$ create a classical one-time pad on measurement outputs.
\item Prover entangles all received qubits according to the $2\kappa+1$ cluster states, each of dimension $n \times m$, by applying $cZ$ gates for every edge of each cluster.
\item Prover measures all qubits simultaneously in angles $\delta_i$ and returns the measurement results to the verifier.
\item Verifier applies a bit flip to the output bit $i$ when $r_i =1$ and to its (non-dummy) neighbours when $r'_i=1$ to undo the classical one time pad. The output string $\boldsymbol{x}$ of the measurements  of the target computation is the output of the protocol. The verifier sets an extra bit to accept if all the traps give the correct result (decoded measurement result $0$).
\end{enumerate}

\vspace{0.3cm}

A variation of the protocol can have the prover to entangle all the graphs directly in the `extended brickwork state' instead of the square lattice state. This leaks no extra more information to the prover from what is publicly known. However, we seek  a more generic prover and keep the protocol as presented.

The resource count of the protocol is as follows. The number of qubits prepared by the verifier and sent to the prover one at a time is $(2\kappa+1)N$ where $N$ is the original size of the computation. The classical information exchanged is linear in $N$ and can be sent in one round. Similarly the classical outcomes of the  measurements can be sent in one go. The prover is  required to entangle all neighbouring qubits in a square lattice and apply single qubit measurements in the $xy$-plane.

\subsection{Proof of Completeness}
\label{sec:correct}

\begin{figure}[t]
\centering
\begin{adjustbox}{max width=\columnwidth}
\begin{tikzpicture}[control/.style={circle,fill,inner sep=0pt,minimum size=0.2cm},target/.style={draw,circle,cross}]
	\foreach \i/\ilbl in {0.8/1,1.6/2,2.4/3,3.2/4,4/5,4.8/6,7/7,8/8} {
		\coordinate (x\ilbl) at (\i,0);
	}

	\foreach \y/\ylbl in {0/1,-3/2,-4/3,-5/4,-8/5,-9/6,-11/7,-12/8}{
		\coordinate (mode-\ylbl) at (0,\y);
	}
	\node (bigU) [draw,minimum height=13cm,minimum width=1.0cm] at ($($(mode-1)!0.5!(mode-8)$)+(6.2,0)$) {$U_B$};

	\node [left] at (mode-1) {$\ket{ +_{\theta_1} }$};
	\node [left] at (mode-2) {$\ket{d_{i-1} } $};
	\node [left] at (mode-3) {$Z^{d_{i+1}} Z^{d_{i-1}} \ket{ +_{\theta_{i}} }$};
	\node [left] at (mode-4) {$\ket{d_{i+1} } $};
	\node [left] at (mode-5) {$\ket{+_{\theta_{(2\kappa+1)N}}}$};
	\node [left] at (mode-6) {$\ket{ \delta_1 }$};
	\node [left] at (mode-7) {$\ket{ \delta_{(2\kappa+1)N} }$};
	\node [left] at (mode-8) {$\ket{0}^{\otimes|B|}$};

	\draw (x1 |- mode-2) -- (x1 |- mode-4);
	\foreach \start/\sign in {1/-,2/+,4/-,5/+} {
		\draw (x1 |- mode-\start) --++(0,\sign 0.4);
		\draw [dashed] ($(x1 |- mode-\start)+(0,\sign 0.4)$) --++ (0,\sign 0.8);
	}

	\foreach \m in {1,2,3,4,5} {
		\node [control] at (x1 |- mode-\m) {};
	}
	\foreach \m in {6,7,8} {
		\node at (x1 |- mode-\m) {\large /};
	}

	\foreach \r [evaluate=\r as \s  using {int(\r+1)}] in {1,2,3,4,5} {
		\node (R-\r) [draw,fill=white,rectangle] at (x\s |- mode-\r) {$R_z$};
		\draw (mode-\r) -- (R-\r.west);
		\draw (R-\r.east) -- (R-\r -| bigU.west);
	}

	\draw (R-1.south) -- (R-1.south |- mode-6) node [control] {};
	\draw (R-5.south) -- (R-5.south |- mode-7) node [control] {};
	\foreach \r in {2,3,4} {
		\draw (R-\r.south) --++(0,-0.4);
		\draw [dashed] ($(R-\r.south)+(0,-0.4)$) --++ (0,-0.8);
	}

	\foreach \y in {1,2,3,4,5} {
		\node [draw,rounded rectangle,rounded rectangle left arc=none,minimum width=1cm] (det-\y) at ( 7.8,0 |- mode-\y) {X};
		\draw ( mode-\y -| bigU.east) -- (det-\y.west);
		\draw (det-\y.7) --++ (0.3,0);
		\draw (det-\y.-7) --++ (0.3,0);
	}
	\foreach \y in {6,7,8} {
		\draw ( mode-\y) -- (mode-\y -| bigU.west);
		\draw ( mode-\y -| bigU.east) -- ( 7.8,0 |- mode-\y);
	}
	
\end{tikzpicture}
\end{adjustbox}

\caption{The inputs, other than prover's private system $\ket{0}^{\otimes|B|},$ are the qubits prepared by the verifier in steps 1-3 of Protocol 1, for both target and trap rounds. We represent the prover's operation (steps 4-6 in Protocol 1) upon their receipt. Qubit at position $i$ is a trap qubit surrounded by dummy qubits at positions $i-1$ and $i+1$.   
$U_B$ is an arbitrary unitary deviation on the prover's system. When prover is honest $U_B=I$. }
\label{completeness1}
\end{figure}
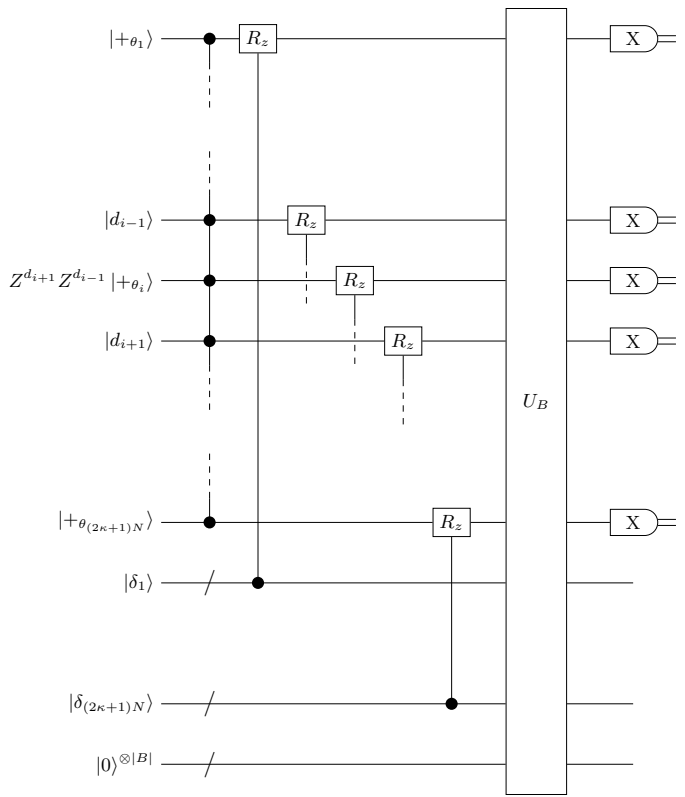

To prove completeness we assume that the prover honestly follows the prescribed steps (up to bounded noise). Before considering the noisy case, we show that for the noiseless prover, the fidelity of the target computation and the trap computation to the correct ones are both unity.

 We begin with a circuit diagram of the operations on the prover's side in Fig.~(\ref{completeness1}).
Any measurement by angle $\{\delta_i\}$ for the prover is  mathematically decomposed into a $z$-rotation ($R_z$) controlled by $\delta_i$ and a Pauli $X$ measurement. Without loss of generality, since everything before the measurements is unitary we can assume that even a dishonest prover will apply the correct unitary operators and then chose his deviation $U_B$ on all systems, including his private qubits $\ket{0}^{\otimes |B|}$. Since we are proving completeness in this section, we assume $U_B=I$. The measurement angles $\delta_i$ received by the prover are represented as computational basis multi-qubit states $\ket{\delta_i}$.

The circuit in Fig.~(\ref{completeness1}) can be simplified in a number of ways, resulting in the circuit of Fig.~(\ref{completeness2}). The $cZ$ gates between the dummy qubits and their neighbours cancel the Pauli $Z$ pre-rotation on the neighbours. Also, we can write explicitly the rotation angles on each of the controlled $R_z$ gates and remove the control lines.

\begin{figure}[t]
\centering
\begin{adjustbox}{max width=\columnwidth}
\begin{tikzpicture}[control/.style={circle,fill,inner sep=0pt,minimum size=0.2cm},target/.style={draw,circle,cross}]
	\foreach \y/\ylbl in {0/1,-3/2,-4/3,-5/4,-8/5,-9/6,-11/7,-12/8}{
		\coordinate (mode-\ylbl) at (0,\y);
	}
	\node (bigU) [draw,minimum height=13cm,minimum width=1.0cm] at ($($(mode-1)!0.5!(mode-8)$)+(7.5,0)$) {$U_B$};

	\node [left] at (mode-1) {$\ket{ +_{\theta_1} }$};
	\node [left] at (mode-2) {$\ket{d_{i-1} }$};
	\node [left] at (mode-3) {$\ket{+_{\theta_{i}}}$};
	\node [left] at (mode-4) {$\ket{d_{i+1} }$};
	\node [left] at (mode-5) {$\ket{+_{\theta_{N'}}}$};
	\node [left] at (mode-6) {$\ket{ \delta_1 }$};
	\node [left] at (mode-7) {$\ket{ \delta_{N'} }$};
	\node [left] at (mode-8) {$\ket{0}^{\otimes|B|}$};

	\foreach \m in {1,5} {
		\node [control] at (x1 |- mode-\m) {};
	}
	\foreach \m in {6,7,8} {
		\node at (x1 |- mode-\m) {\large /};
	}

	\foreach \start/\sign in {1/-,5/+} {
		\draw (x1 |- mode-\start) --++(0,\sign 0.4);
		\draw [dashed] ($(x1 |- mode-\start)+(0,\sign 0.4)$) --++ (0,\sign 0.8);
	}

	\coordinate (xop) at (4.0,0);
	\node (op-1) [draw,rectangle] at ( xop |- mode-1) {$R_z (\theta_1 + (-1)^{r_1'} \phi_1 + r_1\pi )$};
	\node (op-2) [draw,rectangle] at ( xop |- mode-2) {$R_z (\theta_{i-1}  + r_{i-1} \pi )$};
	\node (op-3) [draw,rectangle] at ( xop |- mode-3) {$R_z( \theta_{i} + r_i \pi )$};
	\node (op-4) [draw,rectangle] at ( xop |- mode-4) {$R_z (\theta_{i+1}  + r_{i+1} \pi )$};
	\node (op-5) [draw,rectangle] at ( xop |- mode-5) {$R_z (\theta_{N'} + (-1)^{r_{N'}'} \phi_{N'} + r_{N'}\pi )$};

	\foreach \y in {1,2,3,4,5} {
		\draw (mode-\y) -- (op-\y.west);
		\draw (op-\y.east) -- (op-\y -| bigU.west);
		\node [draw,rounded rectangle,rounded rectangle left arc=none,minimum width=1cm] (det-\y) at ( 9.1,0 |- mode-\y) {X};
		\draw ( mode-\y -| bigU.east) -- (det-\y.west);
		\draw (det-\y.7) --++ (0.3,0);
		\draw (det-\y.-7) --++ (0.3,0);
	}
	\foreach \y in {6,7,8} {
		\draw ( mode-\y) -- (mode-\y -| bigU.west);
		\draw ( mode-\y -| bigU.east) -- ( 9.1,0 |- mode-\y);
	}
\end{tikzpicture}
\end{adjustbox}

\caption{When applying the corresponding entangling operations in Fig.~(\ref{completeness1}), dummy qubits at positions $i-1$ and $i+1$ have the effect of isolating their neighbours and cancelling the neighbours' pre-rotations that depend on parameters $d_{i-1},d_{i+1}$ (here the only neighbour depicted is the trap qubit at position $i$). Also, unitary rotations of Fig.~(\ref{completeness1}) are written explicitly. Remember that for dummy and trap qubits angles $\phi$ take value $0$. For clarity of the figure we have used $N'\equiv(2\kappa+1)N$.}
\label{completeness2}
\end{figure}
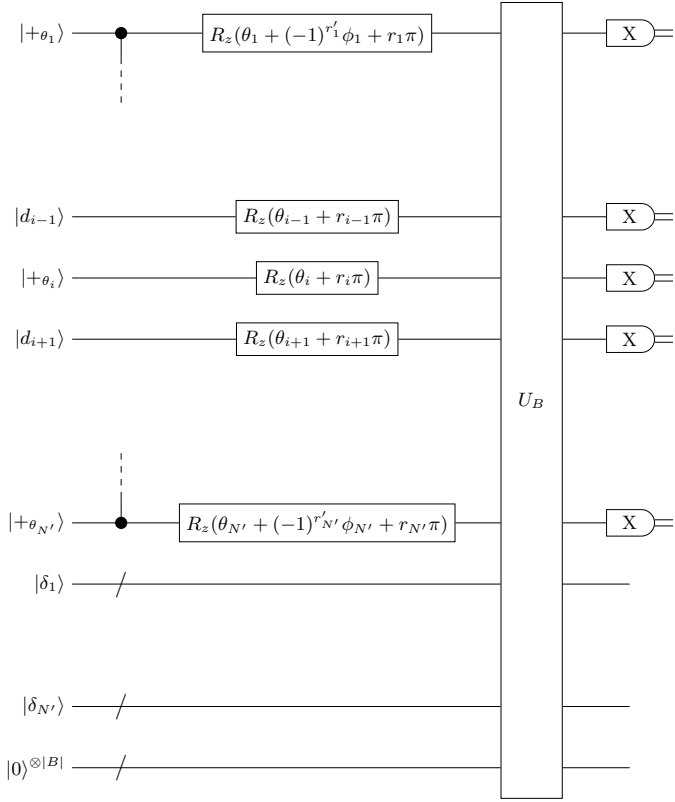

\begin{figure}[h]
\centering
\begin{adjustbox}{max width=\columnwidth}
\begin{tikzpicture}[control/.style={circle,fill,inner sep=0pt,minimum size=0.2cm},target/.style={draw,circle,cross}]
	\foreach \y/\ylbl in {0/1,-3/2,-4/3,-5/4,-8/5,-9/6,-11/7,-12/8}{
		\coordinate (mode-\ylbl) at (0.5,\y);
	}
	\node (bigU) [draw,minimum height=13cm,minimum width=1.0cm] at ($($(mode-1)!0.5!(mode-8)$)+(7.2,0)$) {$U_B$};

	\node [left] at (mode-1) {$\ket{ + }$};
	\node [left] at (mode-2) {$\ket{ 0 }$};
	\node [left] at (mode-3) {$\ket{ + }$};
	\node [left] at (mode-4) {$\ket{ 0 }$};
	\node [left] at (mode-5) {$\ket{ + }$};
	\node [left] at (mode-6) {$\ket{ \delta_1 }$};
	\node [left] at (mode-7) {$\ket{ \delta_{N'} }$};
	\node [left] at (mode-8) {$\ket{0}^{\otimes |B|}$};

	\foreach \m in {1,5} {
		\node [control] at (x1 |- mode-\m) {};
	}
	\foreach \m in {6,7,8} {
		\node at (x1 |- mode-\m) {\large /};
	}

	\foreach \start/\sign in {1/-,5/+} {
		\draw (x1 |- mode-\start) --++(0,\sign 0.4);
		\draw [dashed] ($(x1 |- mode-\start)+(0,\sign 0.4)$) --++ (0,\sign 0.8);
	}

	\coordinate (x2) at (2.0,0);
	\coordinate (x3) at (3.5,0);
	\coordinate (x4) at (5.5,0);
	\node (opR-1) [draw,rectangle] at ( x2 |- mode-1) {$ R_z (\phi_1) $};
	\node (opR-2) [draw,rectangle] at ( x2 |- mode-2) {$ X^{d_{i-1}} $};
	\coordinate (opR-3) at (x2 |- mode-3);
\node (opR-4) [draw,rectangle] at ( x2 |- mode-4) {$ X^{d_{i+1}} $};	
	\node (opR-5) [draw,rectangle] at ( x2 |- mode-5) {$ R_z (\phi_{N'}) $};

	\node (opX-1) [draw,rectangle] at ( x3 |- mode-1) {$ X^{r'_1} $};
	\node (opX-2) [draw,rectangle] at ( x3 |- mode-2) {$ R_z(\theta_{i-1}) $};
	\node (opX-3) [draw,rectangle] at ( x3 |- mode-3) {$ X^{r'_i} $};
	\node (opX-4) [draw,rectangle] at ( x3 |- mode-4) {$ R_z(\theta_{i+1}) $};
	\node (opX-5) [draw,rectangle] at ( x3 |- mode-5) {$ X^{r'_{N'}} $};

	\node (opZ-1) [draw,rectangle] at ( x4 |- mode-1) {$ Z^{ r_1 + \sum_{j \sim 1} r'_j } $};
	\node (opZ-2) [draw,rectangle] at ( x4 |- mode-2) {$ Z^{ r_{i-1}  } $};
	\node (opZ-3) [draw,rectangle] at ( x4 |- mode-3) {$ Z^{r_i} $};	
	\node (opZ-4) [draw,rectangle] at ( x4 |- mode-4) {$ Z^{ r_{i+1} } $};
	\node (opZ-5) [draw,rectangle] at ( x4 |- mode-5) {$ Z^{ r_{ N' } + \sum_{j \sim N'} r'_j } $};

	\foreach \y in {1,2,3,4,5} {
		\draw (mode-\y) -- (opR-\y.west);
		\draw (opR-\y.east) -- (opX-\y.west);
		\draw (opX-\y.east) -- (opZ-\y.west);
		\draw (opZ-\y.east) -- (opZ-\y -| bigU.west);
		\node [draw,rounded rectangle,rounded rectangle left arc=none,minimum width=1cm] (det-\y) at ( 9.1,0 |- mode-\y) {X};
		\draw ( mode-\y -| bigU.east) -- (det-\y.west);
		\draw (det-\y.7) --++ (0.30,0);
		\draw (det-\y.-7) --++ (0.30,0);
	}
	\foreach \y in {6,7,8} {
		\draw ( mode-\y) -- (mode-\y -| bigU.west);
		\draw ( mode-\y -| bigU.east) -- ( 9.1,0 |- mode-\y);
	}
\end{tikzpicture}
\end{adjustbox}

\caption{Each $z$-rotation  by $\theta$ in Fig.~(\ref{completeness2}) undoes the corresponding pre-rotations of the qubits (except for the dummies that have no pre-rotation by $\theta$). For any qubit $k$, operations in the form $R_z((-1)^{r'_k}\phi_k)$ in Fig.~(\ref{completeness2}) can be written as $X^{r'_k} R_z(\phi_k) X^{r'_k}$ and the $X^{r'_k}$ before (in temporal order) when commuting with the entanglement operators can be written as $Z^{r'_k}$ on the neighbours (this has an effect on qubits $1$ and $N'$ in this figure). All Pauli operators here are written separately from $z-$rotations. Notice that we can write an extra $X^{r'_i}$, with $r'_i \in_{R} \{0,1\}$, applying on the trap qubit $i$ since $X\ket{+}=\ket{+}$. }
\label{completeness3}
\end{figure}
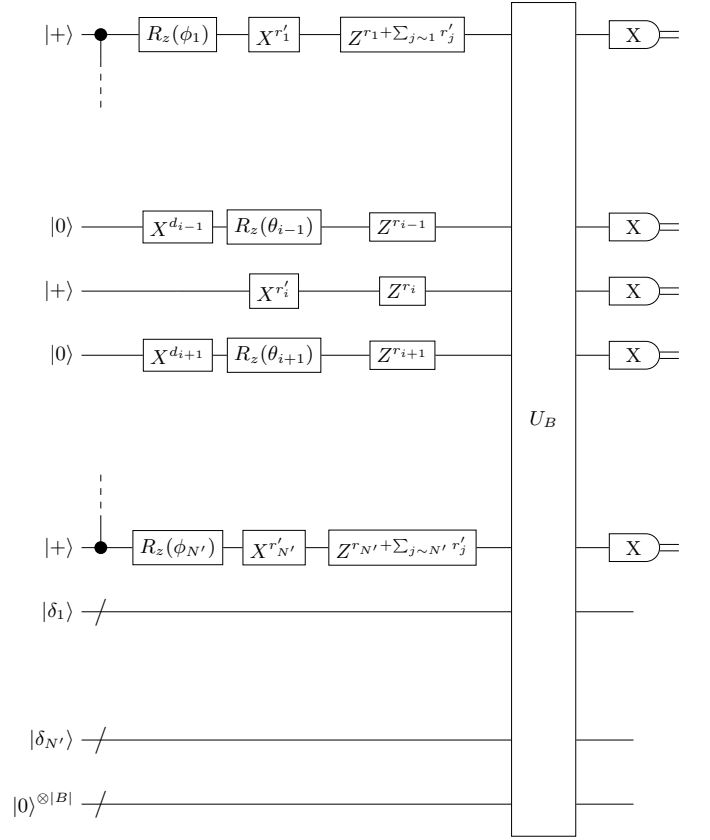

Further simplification follows when the $z$-rotations by $\theta_i$ which are part of the $R_z$ gates and the $z$-rotations by $\theta_i$ applied by the verifier to the qubits before sending them to the prover mutually cancel after commuting with the $cZ$ gates. Notice that the dummy qubits are an exception since $\theta_i$ rotations remain but have no effect other than a global phase. Moreover, we can extract the Pauli operators from the $R_z$ by applying identities: $R_z(-\chi)=X R_z(\chi) X$ and $R_z(\chi+\pi)= Z R_z(\chi)$. The Pauli $X$ operators from the left hand side of $R_z$ can be rewritten as $Z$ rotations on their entangled neighbours. This results in the circuit diagram depicted in Fig.~(\ref{completeness3}). Notice that the remaining Pauli $X$ operators do not have any effect on the Pauli $X$ measurements (we recall that in this proof $U_B=I$) and the Pauli $Z$ operators flip the measurement results.

Let us denote all the measurement outcomes of the protocol except the dummy qubit measurements by the binary vector $\boldsymbol{x}$ and $p(\boldsymbol{x})$ the probability of obtaining it. Let $q^{\text{exc}}(\boldsymbol{x})$ denote the exact probability of obtaining $\boldsymbol{x}$ in an non-encrypted MBQC implementation using the same measurement pattern $\{\phi\}_i^N$ as input. The only difference between the actual and the non-encrypted case are the Pauli $Z$ operators before the measurements, which flip the outcomes. Therefore, by relabelling the probabilities $p(\boldsymbol{x})$ to $p(\boldsymbol{x}')$, where $x'_i = x_i \oplus r_i \oplus \sum_j r'_{j \sim i}$, we get $q^{\text{exc}}(\boldsymbol{x})=p(\boldsymbol{x}')$. In other words, in the noiseless case, we can sample from the exact distribution by simply correcting the bit flips caused by the random Pauli $Z$, which are known to the verifier.

In MBQC, we can also write the distribution in terms of unitaries, labelled by the measurement outcomes (the so-called branches of the MBQC computation) of all the layers except the last. For dimension $m\times n,$ we have (up to global phases)
\begin{flalign}
 &q^{\text{exc}}(\boldsymbol{x})= \nonumber \\& \frac{|\bra{x_{(n-1)m+1}, \ldots, x_{nm}} U_{x_1, \ldots, x_{(n-1)m}} \ket{+_1, \ldots, +_m}|^2}{2^{(n-1)m}} \label{eqnequi}
\end{flalign}
since all the computational branches $(x_1, \ldots, x_{(n-1)m})$ are equiprobable and they define a unitary operation on the input~\cite{childs2005unified}. For the trap qubits this distribution is deterministic since each qubit is prepared in the $\ket{+}$ state, remains isolated throughout the computation and is measured in the $\ket{\pm}$ basis.

Now, consider local bounded noise of the form of Eqn.~(\ref{noise_eq}) after every elementary operation  $j$, including preparation, entangling and measurement. The operations that can introduce noise in a single round of the protocol include $N$ preparations at the verifier's end and at most $2N$ entanglements and $N$ measurements at the prover's end. This is an upper bound of $4N$ operations. The fidelity $F_c^2$ of the noisy output of the target computation to the noiseless one (which is the correct one as we proved above) cannot be smaller than $1 - (N \epsilon_V + 3N \epsilon_P)$. Since for any two states $\rho,\sigma$, $D(\rho,\sigma)\leq \sqrt{1-F^2(\rho,\sigma)}$, this is an upper bound in total variation distance for the target computation  $1-\delta = \sqrt{N (\epsilon_V + 3 \epsilon_P)}$.

Completeness means that our scheme should  accept with high probability in the case of bounded noise. The acceptance of the scheme, according to Def.~(\ref{def_sch}), depends on our estimate $\widehat{F}^2_t$ of the acceptance probability of the protocol $F^2_t$. Given the above bounded noise,  $F^2_t$  cannot be smaller than $1 - \kappa N (\epsilon_V + 3 \epsilon_P)$.

Our estimate for $F^2_t$ comes from $M$ i.i.d. repetitions of the protocol.  By Hoeffding's inequality, repeating $M = \log(1/\beta)/(2 \kappa^2 N^2 (\epsilon_V+ \epsilon_P)^2)$ times gets us $\kappa N (\epsilon_V + \epsilon_P)$-close in our estimation with confidence $1 - \beta$.
In order to have high probability of acceptance we need to set the limit for accepting the estimate to $(1-  \kappa (2 N\epsilon_V + 4 N \epsilon_P))$. Then our probability of accepting is as high as our confidence. Setting this limit is necessary to get high completeness but will have an effect in the soundness.

\subsection{Proof of Soundness}
\label{sec:sound}

The proof of soundness of Theorem \ref{thm1} is based on the fact that the fraction of accepting protocols, $\widehat{F}^2_t$,  is a good estimator of a lower bound in the fidelity $F^2_c$ of the target computation. Thus, looking at $\widehat{F}^2_t$ gives us with high confidence a lower bound on the fidelity, or similarly an upper bound on total variation distance var, as defined in Eq.~(\ref{eq:var_def}).

We outline the main arguments employed to prove this theorem in stages here, and provide the explicit algebraic derivations in Appendix~\ref{app:Thm1}.

Firstly, a unitary deviation $U_B$ applied before the measurements, depicted in Fig.~(\ref{completeness3}), captures in all generality the prover's dishonesty.  To see this, consider the case when the prover performs measurements different from the honest ones. This corresponds to applying a unitary basis rotation followed by Pauli $X$ measurements. Then, $U_B$ applies also on the prover's private subsystem so he can use this power to replace the qubits he receives with any other qubits he chooses to prepare privately. In any case, he has to report some classical measurement results so we always keep the final Pauli $X$ measurements in the picture. Our proof should therefore apply to any choice of $U_B$.

Secondly, we bound the total variation distance of the output distribution via the trace distance $D(\rho_{c},\rho'_c)$, where $\rho_c$ represents the state of the computational system just \emph{after} the Pauli $X$ measurements if the prover is honest and $\rho'_c$ the same state if the prover is dishonest. Thus,
\begin{eqnarray}
\text{var} \leq D(\rho_c,\rho'_c) \leq \sqrt{1-F^2(\rho_c,\rho'_c)} \nonumber \\
= \sqrt{1 - \text{Tr}^2(\sqrt{\rho_c \rho'_c}) }
\end{eqnarray}

The main idea leading to the statement of the theorem is that the acceptance probability $F^2_t$ minus a lower bound on the fidelity of the computational system $F^2(\rho_c, \rho'_c)$ is small, when averaged over the random parameters. 
Therefore, by estimating $F^2_t$ (by counting the fraction of acceptances over many repetitions of the protocol), we get a good estimate of a lower bound on the fidelity of the computational system and therefore an upper bound on var. 
We begin our analysis for the case of perfect preparations and subsequently  incorporate the effect of noise.

Averaged over the random parameters, the probability of getting all trap outcomes $1$, summing over the random variables $r_i, r'_i, d_i$ and $\theta_i$, is calculated in Appendix~\ref{app:Thm1} as

\begin{equation}
F_t^2 = \sum_{\boldsymbol{t}} p(\boldsymbol{t})   \sum_{k} |\alpha_k|^2  \prod_{i \in \boldsymbol{t}}  |\bra{+}_i    P_{k |i }  \ket{+}_i |^2,
\label{eq1}
\end{equation}
where $\boldsymbol{t}$ is the vector of the indices of the positions of the traps in all $2\kappa$ trap systems and $\sum_{\boldsymbol{t}}$ takes all possible values allowed by the construction with equal probability $p(\boldsymbol{t})$. The summation over the random parameters results in the attack on the trap system to be transformed into a convex combination of Pauli operators $P_{k}$, each with probability $|\alpha_k|^2$. By $P_{k|i}$ we represent the Pauli operator  that applies on qubit $i$.

The average fidelity $F_c$ of the computational system is

\begin{eqnarray}
F_c & \equiv & \sum_{\boldsymbol{r},\boldsymbol{\theta},\boldsymbol{t}} p(\boldsymbol{r},\boldsymbol{\theta},\boldsymbol{t}) F(\rho_c, \rho'_c)  \\ & = & \sum_{\boldsymbol{r},\boldsymbol{\theta},\boldsymbol{t}} p(\boldsymbol{r},\boldsymbol{\theta},\boldsymbol{t})  \text{Tr}( \sqrt{\rho_c \rho'_c} ), 
\end{eqnarray}
where $\rho_c$ and $\rho'_c$ represent the honest and dishonest state of the target computation just after the Pauli $X$ measurements. Calculation, presented in detail in Appendix~\ref{app:Thm1}, leads to
\begin{equation}
F_c^2 \geq \sum_{\boldsymbol{t}} p(\boldsymbol{t}) \sum_{k} |\alpha_{k}|^2  \prod_{i \in \boldsymbol{c(t)}}  | \bra{+}_i     P_{k|i}    \ket{+}_i |^2 
\label{eq2} 
\end{equation}
where $\boldsymbol{c(t)}$ denotes the positions of the qubits that participate in the computation and depends on the random ordering of the $2\kappa+1$ rounds and therefore is a function of the position of the traps.

In general we prove that

\begin{equation}
 F_t^2 -  F_c^2  \leq \frac{\kappa! (\kappa+1)!}{(2\kappa+1)!} \equiv \Delta_{\kappa}. 
\end{equation}

The verification scheme output bit is set to accept or reject by averaging over $M$ repetitions of the protocol and comparing our estimate of  $F_t^2$ with  $(1-  \kappa N(2 \epsilon_V + 4 \epsilon_P))$ (set by completeness). By Hoeffiding's inequality repeating $M = \log(1/\beta)/(2 \kappa^2 N^2 (\epsilon_V + \epsilon_P)^2)$ times gets us $\kappa N (\epsilon_V + \epsilon_P)$-close in our estimation with confidence $1 - \beta$. Therefore, $F_t^2 -   F_c^2  \leq  \kappa N(3 \epsilon_V + 5 \epsilon_P) + \Delta_{\kappa}.$  This means that for the total variation distance we have an upper bound, which gives the soundness parameter $\varepsilon$ of Def.~(\ref{def_ver})
 $$\text{var} \leq \varepsilon = \sqrt{\kappa N (3 \epsilon_V + 5 \epsilon_P) + \Delta_{\kappa}}.$$

\section{Fault-tolerant Verification of Ising Sampler}
\label{sec:FTver}

Ensuring $N \epsilon_V$ and $N \epsilon_P$ in Theorem~\ref{thm1} to be constant  will get harder experimentally for increasing $N$. 
Therefore, we present two new fault-tolerant verification schemes where the total noise scales linearly with the size, and prove that it provides a distribution that is hard to sample from classically upto constant additive error. We then prove that noise scaling with system size does not prevent us from verifying the prover’s distribution with completeness and soundness parameters independent of the problem size.

Quantum fault tolerance strategies such as due to RHG~\cite{raussendorf07rhg_topological} can overcome the challenge of noise scaling with system size. This involves gate distillation requiring adaptive operations which are beyond the Ising sampler. On the target computation, our fault tolerant verification schemes overcome this adaptivity by using arguments for free postselection due to Fujii~\cite{fujii2016noise} as applied to the verification of quantum supremacy.
On the trap computation, we do not require any adaptivity since we chose it to be Clifford. This keeps our fault tolerant verification schemes within the Ising sampler, allowing verification of quantum supremacy in the presence of total noise scaling linearly with the size. Note that a non-Clifford trap computation would suffer due to nonadaptivity. Time complexity of the quantum operations in the protocol is constant and the number of qubits needed is $O(N \mathrm{PolyLog}(N))$~\cite{raussendorf07rhg_topological}, the polylogarithmic overhead coming from the properties of the topological code and the use of concatenation in the distillation procedure.

The next issue of fault-tolerant thresholds leads to two fault-tolerant versions of the protocol in Fig.~(\ref{protocolss}) and described in detail in Sections~\ref{sub_prot_ft_a} and \ref{sub_prot_ft_b}. The first is called Protocol 2a. It employs the full RHG encoding in the traps leading to the threshold of  $\epsilon_{\text{thres}}=0.75\%$~\cite{raussendorf07rhg_topological}, the same threshold as for universal quantum computation.
	This is worse than the suggested improvements in the noise thresholds for unverified quantum supremacy~\cite{fujii2016noise}. 
	However, our next protocol, Protocol 2b, provides $\epsilon_{\text{thres}} = 1.97\%$ for verified quantum supremacy, which is an underestimate because of the analytical treatment and could potentially be improved by numerical simulation as in Ref.~\cite{raussendorf07rhg_topological}. 
	
To achieve this threshold, Protocol 2b, replaces error correction with error detection when performing the RHG encoding on the trap qubits. This is possible because the trap qubits are isolated and can be retransmitted individually without affecting the rest of the trap computation.
The numerical value is obtained by performing a threshold calculation of applying the RHG encoding in MBQC  (Appendix~\ref{threshh}). A similar procedure was performed in the circuit model by Fujii~\cite{fujii2016noise}. 
The cost of maintaining the same completeness and soundness as in Protocols 1 and 2a is to replace $\kappa$ in Fig.~(\ref{protocolss}) by $M\kappa,$ where $M$ is an extra overhead in the number of qubits depending on the code minimal distance $d$ between and around the defects and the noise parameters $\epsilon_{V}$ and $\epsilon_{P}$.  For example, with $d=2$ and $\epsilon_{V} = \epsilon_{P} = \epsilon$ as the following fractions of  the noise threshold, we have

\vspace{0.3cm}

\begin{center}
\begin{tabular}{ | m{2em} | m{5em} | m{5em}| m{4.5em} | } 
\hline
$\epsilon$ &$\epsilon_{\text{thres}}/20$ & $\epsilon_{\text{thres}}/50$ & $\epsilon_{\text{thres}}/100$ \\ 
\hline
$M$ &$3 \times 10^8$ & $2863$ & $54$ \\  
\hline
\end{tabular}
\end{center}

\vspace{0.3cm}

Improvement in $M$ may also be possible with judicious braiding or using an alternative topological code.

An additional intricacy needs resolving for both fault-tolerant protocols. Since blindness is an ingredient in our verification scheme, its straightforward application (on the logical level) risks leaking the logical measurement angles in the distillation procedure, where many copies of the same magic state need to be sent. Also, for the distillation procedure to be effective, we need to reveal information about the state distilled. Our stratagem for circumventing this is to apply blindness on the lowest level of MBQC, on which the fault-tolerant construction is based. The traps are applied at the logical MBQC level, since those are the qubits needing protection from noise, as outlined in Fig.~(\ref{layers1}).

\begin{figure}[h]

 \begin{adjustbox}{max width=1\columnwidth}
\begin{tikzpicture}[every node/.style={%
	rectangle,
	rounded corners=0.4cm,
	align=center,
	draw,
	thick,
	minimum height=2cm,
	shade,
	shading=radial,
	inner color=lightgray!20,
	outer color=lightgray!50,font=\large
}]
	\node [minimum width=13.7cm] at (0,2.10) {Ising Sampler and Trap Computations MBQC \\ (Logical layer)};
	\node [minimum width=13.7cm] at (0,0) {Protected topology using defects};
	\node [minimum width=13.7cm] at (0,-2.10) {Blind 3D cluster-state MBQC \\ (Physical layer)};
\end{tikzpicture}
\end{adjustbox}

\caption{Layered structure of verifiable FT computation.}
\label{layers1}
\end{figure}
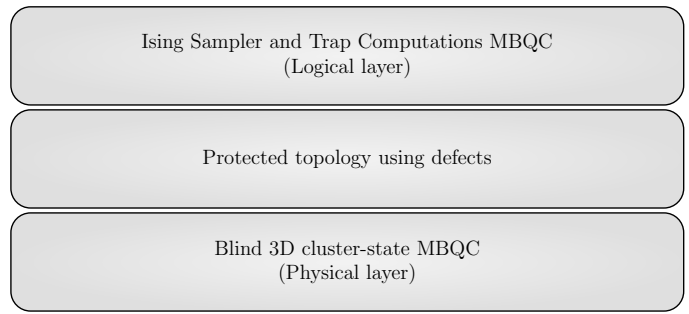

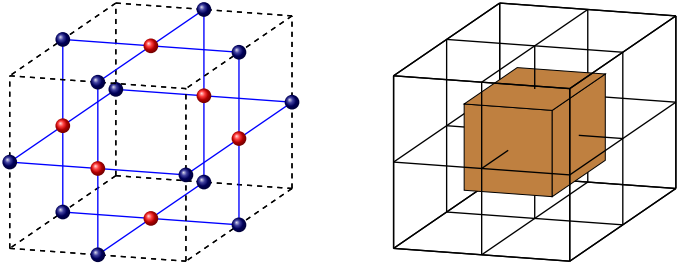
\begin{figure}[b!]

\centering
 \begin{adjustbox}{max width=\columnwidth,center}
\begin{tikzpicture}[every edge/.append style={thick},y={(0,0.9)},rotate around y=-10]
	\begin{scope}[rotate around x=-00]
	\foreach \a in {-2,0,2} {
		\draw (\a,-2,-2) edge (\a,2,-2);
		\draw (-2,\a,-2) edge (2,\a,-2);
		\draw (-2,-2,\a) edge (-2,2,\a);
		\draw (-2,\a,-2) edge (-2,\a,2);
		\draw (\a,-2,-2) edge (\a,-2,2);
		\draw (-2,-2,\a) edge (2,-2,\a);
	}
	\draw (-2,0,0) edge (0,0,0);
	\draw (0,-2,0) edge (0,0,0);
	\draw (0,0,-2) edge (0,0,0);
	\draw [fill=brown] (-1,-1,1) -- (-1,1,1) -- (1,1,1) -- (1,-1,1) -- cycle;
	\draw [fill=brown] (1,-1,-1) -- (1,1,-1) -- (1,1,1) -- (1,-1,1) -- cycle;
	\draw [fill=brown] (-1,1,-1) -- (-1,1,1) -- (1,1,1) -- (1,1,-1) -- cycle;
	\foreach \a in {-2,0,2} {
		\draw (\a,-2,2) edge (\a,2,2);
		\draw (-2,\a,2) edge (2,\a,2);
		\draw (2,-2,\a) edge (2,2,\a);
		\draw (2,\a,-2) edge (2,\a,2);
		\draw (\a,2,-2) edge (\a,2,2);
		\draw (-2,2,\a) edge (2,2,\a);
	}
	\draw (2,0,0) edge (1,0,0);
	\draw (0,2,0) edge (0,1,0);
	\draw (0,0,2) edge (0,0,1);
\begin{scope}[xshift=-8cm,
	blueball/.style={circle,shading=ball,ball color=Navy},
	redball/.style={circle,shading=ball,ball color=red}
]
	\foreach \a/\astyle in {-2/dashed,0/blue,2/dashed} {
		\draw [\astyle] (\a,-2,-2) edge (\a,2,-2);
		\draw [\astyle] (-2,\a,-2) edge (2,\a,-2);
		\draw [\astyle] (-2,-2,\a) edge (-2,2,\a);
		\draw [\astyle] (-2,\a,-2) edge (-2,\a,2);
		\draw [\astyle] (\a,-2,-2) edge (\a,-2,2);
		\draw [\astyle] (-2,-2,\a) edge (2,-2,\a);
	}
	\node [redball] at (0,0,-2) {};
	\node [redball] at (0,-2,0) {};
	\node [redball] at (-2,0,0) {};
	\node [blueball] at (-2,-2,0) {};
	\node [blueball] at (0,-2,-2) {};
	\node [blueball] at (-2,0,-2) {};
	\foreach \a/\astyle in {-2/dashed,0/blue,2/dashed} {
		\draw [\astyle] (\a,-2,2) edge (\a,2,2);
		\draw [\astyle] (-2,\a,2) edge (2,\a,2);
		\draw [\astyle] (2,-2,\a) edge (2,2,\a);
		\draw [\astyle] (2,\a,-2) edge (2,\a,2);
		\draw [\astyle] (\a,2,-2) edge (\a,2,2);
		\draw [\astyle] (-2,2,\a) edge (2,2,\a);
	}
	\node [redball] at (0,0,2) {};
	\node [redball] at (0,2,0) {};
	\node [redball] at (2,0,0) {};
	\node [blueball] at (-2,2,0) {};
	\node [blueball] at (2,-2,0) {};
	\node [blueball] at (2,2,0) {};
	\node [blueball] at (0,-2,2) {};
	\node [blueball] at (0,2,-2) {};
	\node [blueball] at (0,2,2) {};
	\node [blueball] at (-2,0,2) {};
	\node [blueball] at (2,0,-2) {};
	\node [blueball] at (2,0,2) {};
\end{scope}
\end{scope}
\end{tikzpicture}
 \end{adjustbox}

\caption{3D cluster state used in the RHG code using MBQC. Blue dots are the qubits that represent the primal cubic lattice edges (or equivalently the dual cubic lattice faces) and red dots are the qubits that represent the primal cubic lattice faces (or equivalently the dual cubic lattice edges). Entangling operations ($cZ$) are represented by blue lines. On the right hand side you can see the primal and dual cubes, as are adapted from Refs.~\cite{raussendorf07rhg_topological,
fowler09simplified_topological_ft}.}
\label{3dgraph}
\end{figure}

\begin{figure}[t]
\centering
\begin{minipage}[c]{\linewidth}
\Qcircuit  @C=0.8em @R=0.8em {
	&	&	& \lstick{\ket{+}}& \ctrl{14}	& \qw		& \qw		& \qw		& \qw		& \gate{T}	& \qw & \meter & X \\
		&	&	& \lstick{\ket{+}}& \qw		& \ctrl{13}	& \qw		& \qw		& \qw		& \gate{T}	& \qw & \meter & X \\
		&	&	& \lstick{\ket{0}}& \qw		& \qw		& \targ		& \targ		& \targ		& \gate{T}	& \qw & \meter  & X \\
		&	&	& \lstick{\ket{+}}& \qw		& \qw		& \ctrl{11}\qwx	& \qw\qwx	& \qw		& \gate{T}	& \qw 	& \meter  & X \\
		&	&	& \lstick{\ket{0}}& \qw		& \targ		& \qw		& \targ\qwx	& \targ		& \gate{T}	& \qw 	& \meter  & X \\
		&	&	& \lstick{\ket{0}}& \targ	& \qw		& \qw		& \targ\qwx	& \targ		& \gate{T}	& \qw 	& \meter  & X \\
		&	&	& \lstick{\ket{0}}& \targ	& \targ		& \targ		& \qw\qwx	& \qw		& \gate{T}	& \qw 	& \meter  & X \\
		&	&	& \lstick{\ket{+}}& \qw		& \qw		& \qw		& \ctrl{7}\qwx	& \qw		& \gate{T}	& \qw 	& \meter  & X \\
		&	&	& \lstick{\ket{0}}& \qw		& \targ		& \targ		& \qw		& \targ		& \gate{T}	& \qw 	& \meter  & X \\
		&	&	& \lstick{\ket{0}}& \targ	& \qw		& \targ		& \qw		& \targ		& \gate{T}	& \qw 	& \meter  & X \\
		&	&	& \lstick{\ket{0}}& \targ	& \targ		& \qw		& \targ		& \qw		& \gate{T}	& \qw 	& \meter  & X \\
		&	&	& \lstick{\ket{0}}& \targ	& \targ		& \qw		& \qw		& \targ\qw	& \gate{T}	& \qw 	& \meter  & X \\
		&	&	& \lstick{\ket{0}}& \targ	& \qw		& \targ		& \targ		& \qw		& \gate{T}	& \qw 	& \meter  & X \\
		&	&	& \lstick{\ket{0}}& \qw		& \targ		& \targ		& \targ		& \qw		& \gate{T}	& \qw 	& \meter  & X \\
\lstick{\ket{+}}& \ctrl{1}	& \qw	& \qw	& \targ		& \targ		& \targ		& \targ		& \ctrl{-12}	& \gate{T}	& \qw 	& \meter  & X \\
\lstick{\ket{0}}& \targ		& \qw	& \qw	& \qw		& \qw		& \qw		& \qw		& \qw		& \qw		& \qw	& \qw
}
\end{minipage}

\caption{Distillation step \cite{bravyi05gate_teleportation}. A logical $\ket{+}$ is produced using the $(15,1,3)$ quantum  Reed-Muller code. Then a transversal $T$ gate applied using the technique of Figure (\ref{teleport1}), stabilizer measurements and teleportation to an auxiliary qubit gives a `cleaner' magic state (up to Pauli $Z$ correction on the teleported qubit) when the error syndromes are correct. Everything is topologically protected. Picture adapted from Ref.~\cite{fujii2013quantum}.}
\label{distil1}
\end{figure}
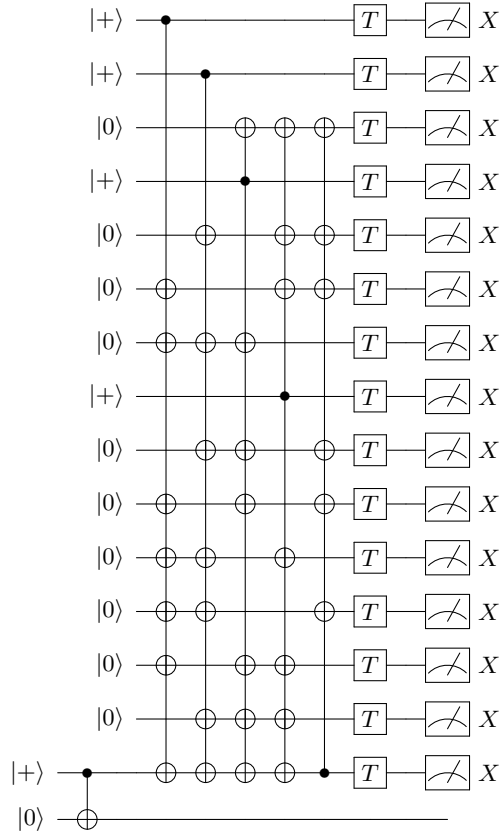

Our proof of the classical hardness of Ising sampling in this case (Theorem \ref{thm3}) relies on proving the completeness and soundness of verifying conditional probabilities (Theorem \ref{thm2}). They are proved in Section~\ref{sub_prot_ft_a}.

Noise is again of the form of Eqn.~(\ref{noise_eq}). Suppose the verifier's noise in each qubit preparation is local, bounded by $\epsilon_V < \epsilon_{\mathrm{thres}},$ the threshold and does not depend on the secret parameters. Assume the honest prover's noise in each elementary operation is bounded by $\epsilon_P< \epsilon_{\mathrm{thres}}$.  In order to prove Theorem \ref{thm2} we make the extra assumption that verifier's noise  is independent of the secret parameters.

In the following theorem, $\epsilon''$ is the error rate of the code and scales down exponentially with distance parameter $d$. Let $q^{\text{nsy}}(\bm{x}|y=0)$ be the  experimental  and $q^{\text{exc}}(\bm{x}|y=0)$ the exact distribution  of the Ising sampler, when they are conditioned on the syndrome measurement outcome $y$ giving the null result. The theorem holds for both Protocol 2a and 2b.

\vspace{0.7cm}
\begin{theorem}[Fault-tolerant verification scheme]
\label{thm2}
There exists a verification scheme with Protocol 2a/2b, $M=\log(1/\beta)/(2 \epsilon''^2)$ and $l=(1- 2 \epsilon'')$, that according to Def.~(\ref{def_ver_post}), which is based on the variation distance between conditional probabilities $q^{\text{nsy}}(\bm{x}|y=0)$ and $q^{\text{exc}}(\bm{x}|y=0)$, is
$$(1-\beta,1-\sqrt{\epsilon''})-\text{complete}$$ and $$(1-\beta,\sqrt{3\epsilon'' + \Delta_{\kappa}})-\text{sound}$$
 where $\Delta_{\kappa}=  \kappa! (\kappa+1)!/(2\kappa+1)!$.
\end{theorem}

\subsection{Protocol 2a}
\label{sub_prot_ft_a}

The fault-tolerant computation scheme used is the one proposed by RHG  \cite{raussendorf07rhg_topological}. Single qubit preparation/distillation, entangling gates ($cX$) and Pauli $X$ measurements are topologically protected using the three dimensional lattice shown in Fig.~(\ref{3dgraph}) and measurement-based implementation of the topological operations (more details on the MBQC implementation of the RHG code see \cite{fowler09simplified_topological_ft}, \cite{fujii2015quantum}). Universality comes from topologically protected concatenated distillation (Fig.~(\ref{distil1})) of magic states:
 $$\ket{Y} = \frac{1}{\sqrt{2}} (\ket{0} + i \ket{1}), \ket{A} = \frac{1}{\sqrt{2}} (\ket{0} + e^{i\frac{\pi}{4}} \ket{1})$$
which are generated by single physical qubit measurements. Using the logical distilled magic states and the gate teleportation model one can implement a universal set of gates (Fig.~\ref{teleport1}). One can simulate an MBQC computation by using these gates and Pauli measurements and consequently add a forth dimension to the system, which comes from the flow of the logical MBQC operations (notice that this layer is distinct from the physical MBQC layer on which the topological code is implemented). The exact usage of RHG encoding in our FT verification scheme depends on whether we use error correction or error detection in our trap computation, giving two separate protocols.

Our first fault-tolerant protocol follows the protocol for the non-fault-tolerant verification of quantum supremacy, introduced in Sec. \ref{main:nft} but using fault tolerance.

\vspace{0.4cm}

\textbf{Protocol 2a:}

\begin{enumerate}
\item \textbf{Trapification:} Verifier selects a random ordering of $2\kappa+1$ sufficiently large 3D graphs of Fig.~(\ref{3dgraph}), one for the target computation and $2\kappa$ for the trap computations. In the target computation round the logical computation is the same as the one in the non-fault-tolerant protocol (see Fig.~(\ref{brick1})). In the trap computation rounds, the logical graph contains isolated traps in the same configuration as in the non-fault-tolerant version (see Fig.(~\ref{protocol_pic})). We call this the logical layer of our protocol.

\item \textbf{Generation of the `topological code-compatible' circuit:} The above MBQC patterns contain $\ket{\pm_{\phi}}$ measurements that are not compatible with the topological code. We directly translate the MBQC patterns into a circuit with the same operations, with the difference that the measurements are replaced by teleportation of distilled $\ket{+_{\phi}}$ states  followed by Pauli $X$ measurements. Since our rotations are multiples of $\pi/8$, $\ket{A}$ magic states need be replaced by $\frac{1}{\sqrt{2}}(\ket{0}+e^{i \frac{\pi}{8}}\ket{1})$ states, gate teleportation is as described before (Fig.~(\ref{teleport1})) with a $T$ gate  instead of a $S$ gate correction and distillation based on $(31,1,3)$  instead of $(15,1,3)$ quantum  Reed-Muller code.  Adaptive $T$ gates are not applied since we want to keep the model instantaneous - this will be accounted for in our supremacy proof.
To avoid adaptivity in distillation we fix which magic states we keep for the next level of distillation independently of the syndrome measurements outcomes - this is not a problem because we have `free' postselection on the syndrome measurements of the target computation (register $y$).

Even the $\ket{\pm}$ measurements of the traps use inputs that go through the distillation and teleportation procedure (Fig.~(\ref{teleport1}) (iii)). This is in order for the physical attacks to have the same effect  on the target and trap computation at the logical level (see proof of verifiability for more).

\item \textbf{Topological translation:} The topological translation from the circuit to the topology is straightforward~\cite{fowler09simplified_topological_ft}.

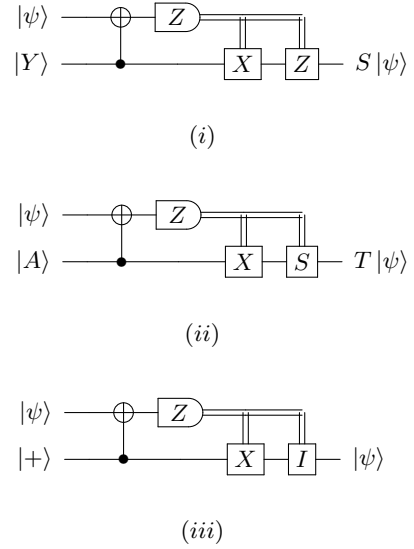
\begin{figure}[t]
\centering
\[
\Qcircuit @C=1em @R=.7em {
\lstick{\ket{\psi}} & \qw & \targ  & \measureD{Z}  & \cw \cwx[1] & \cw \cwx[1] \\
\lstick{\ket{Y}} & \qw & \ctrl{-1} & \qw & \gate{X} & \gate{Z} & \rstick{S\ket{\psi}} \qw
}
\]

\vspace{2ex}

$(i)$

\[
\Qcircuit @C=1em @R=.7em {
\lstick{\ket{\psi}} & \qw & \targ  & \measureD{Z}  & \cw \cwx[1] & \cw \cwx[1] \\
\lstick{\ket{A}} & \qw & \ctrl{-1} & \qw & \gate{X} & \gate{S} & \rstick{T\ket{\psi}} \qw
}
\]

\vspace{2ex}

$(ii)$

\[
\Qcircuit @C=1em @R=.7em {
\lstick{\ket{\psi}} & \qw & \targ  & \measureD{Z}  & \cw \cwx[1] & \cw \cwx[1] \\
\lstick{\ket{+}} & \qw & \ctrl{-1} & \qw & \gate{X} & \gate{I} & \rstick{\ket{\psi}} \qw
}
\]

\vspace{2ex}

$(iii)$

\caption{$(i),(ii)$ Gate teleportation (up to global phase) using magic states $\ket{Y},\ket{A}$, $(iii)$ State teleportation using auxiliary state $\ket{+}$ that mimics gate teleportation (via blindess). These operations are applied in a topologically protected way, both during state distillation using `impure' states and to implement the corresponding logical operations during computation using distilled states.}
\label{teleport1}
\end{figure}

\item \textbf{Blind implementation of topology:}  The topological code is implemented at the physical level by MBQC using the 3D-graphs, so that we can implement them blindly using the following encryption.

\begin{enumerate}

\item Verifier prepares, one by one, the pre-rotated physical qubits $\ket{+_{\theta}}$,  $\theta \in_R \{0, \frac{\pi}{4},\frac{2\pi}{4}, \ldots,  \frac{7\pi}{4}\}$,  needed for the blind implementation of the topological protected computation on the three dimensional cluster states and sends them to the prover. Blindness, induced by the random rotations, hides from the prover the physical operations applied and therefore the logical structure of the computation in the topologically protected (vacuum) and isolated qubit region. In particular, the prover is not able to distinguish between implementing distillation and teleportation of a magic state or a $\ket{+}$ state used for computation and testing respectively.

\item Verifier sends all the encrypted measurement angles $\delta_i = \theta_i + (-1)^{r'_i}\phi_i + r_i \pi$ for $r_i, r'_i \in_R \{0,1\}$. Parameters $r_i, r'_i$ are classical one-time pads for the measurement outputs.

\item Prover runs the computation, by entangling, measuring all at once and returning the results.
\item Verifier classically corrects the returned outcomes using the correction procedure of the quantum error correcting code used in distillation and the topological code and undoes the $r,r'$ pad.
\item Verifier accepts the outcome, which is the logical output string $\boldsymbol{x}$ and  syndrome measurement bit
 $y$ ($y=0$ no error, $y=1$ error) of the target computation, if all the results of the logical trap computations are correct, otherwise rejects.
\end{enumerate}
\end{enumerate}

\vspace{0.3cm}
Completeness of the protocol follows the same analysis as in the non-fault-tolerant case. We can eliminate the pre-rotations by the $\theta$'s of the computation in the lower level MBQC due to $\theta$ being in $\delta$ and, then, the computation is  correct up to Pauli $Z$ corrections before the measurements. Local noise is taken care of by the error correction if it is lower than the threshold of the RHG code.  This avoids scaling issues that we had in the non fault-tolerant protocol.
In particular because of fault tolerance we get   $\text{var} \leq  \sqrt{\epsilon''}$. Completeness means that our scheme should also accept with high probability and this is achieved by setting the limit to accept the fidelity estimate to $(1- 2 \epsilon'')$. By repeating $N = \log(1/\beta)/(2 \epsilon''^2)$ times gets us $\sqrt{\epsilon''}$-close in our estimation with confidence $1 - \beta$. Thus, this is a lower bound on the probability our scheme accepts in the case of completeness.

 The proof of soundness is  similar to the non fault-tolerant case since the noise can be considered part of the prover's attack that  has the same effect on the target computation and the trap at a logical level. We show this in Appendix~\ref{app:Thm2}.

The threshold of this protocol is the same as the threshold of the RHG code since error correction is used in the trap rounds.

\subsection{Protocol 2b}
\label{sub_prot_ft_b}

\begin{figure}[t!]
\centering
\begin{subfigure}[t]{.35\columnwidth}
\centering
\Qcircuit  @C=0.8em @R=0.8em {
	&	&	& \lstick{\ket{+}}& \qw & \ctrl{1}	& \qw		& \qw		& \qw		& \meter & X \\
		&	&	& \lstick{\ket{+}}& \qw& \ctrl{-1}		& \ctrl{1}	& \ctrl{3}		& \qw		&  \meter & X \\
		&	&	& \lstick{\ket{+}}& \qw& \qw		& \ctrl{-1}	& \qw		& \qw & \meter  & X \\
		&	&	& \lstick{\ket{+}}& \qw& \ctrl{1}		& \qw		& \qw & \qw	& \meter  & X \\
		&	&	& \lstick{\ket{+}}& \qw& \ctrl{-1}		& \ctrl{1}		& \ctrl{-3}		& \qw	&  \meter  & X \\
		&	&	& \lstick{\ket{+}}& \qw& \qw	& \ctrl{-1}		& \qw		& \qw	& \meter  & X
}
\caption{}
\end{subfigure}
\hspace{1.3cm}
\begin{subfigure}[t]{.35\columnwidth}
\centering
\Qcircuit  @C=0.6em @R=0.8em {
	&	&	& \lstick{\ket{+}}& \qw & \ctrl{1}	& \qw		& \qw		& \qw		& \meter & X \\
		&	&	& \lstick{\ket{0}}& \qw& \targ		& \targ	& \targ		& \qw		&  \meter & Z \\
		&	&	& \lstick{\ket{+}}& \qw& \qw		& \ctrl{-1}	& \qw		& \qw & \meter  & X \\
		&	&	& \lstick{\ket{0}}& \qw& \targ		& \qw		& \qw & \qw	& \meter  & Z \\
		&	&	& \lstick{\ket{+}}& \qw& \ctrl{-1}		& \ctrl{1}		& \ctrl{-3}		& \qw	&  \meter  & X \\
		&	&	& \lstick{\ket{0}}& \qw& \qw	& \targ		& \qw		& \qw	& \meter  & Z
}
\caption{}
\end{subfigure}

\begin{subfigure}{0.5\columnwidth}
\includegraphics[width=\textwidth]{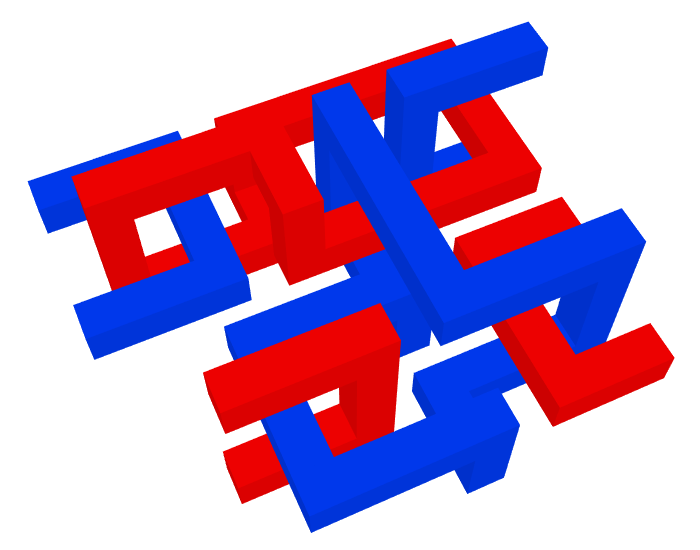}
\caption{}
\end{subfigure}

\vspace{0.3cm}
\begin{subfigure}{0.8\columnwidth}
\centering
\resizebox{\linewidth}{!}{
\begin{tikzpicture}

\pgftransformcm{1}{0}{0.4}{0.5}{\pgfpoint{0cm}{0cm}};
\foreach \x in {10,14,18,22,26} {
	\foreach \y in {3,7} {            

	\node at (\x,\y) [circle,draw=black,thick,fill=red] {};

	}
}

\foreach \x in {12,16,20,24} {
	\foreach \y in {3,7} {            

	\node at (\x,\y) [circle,draw=black,thick,fill=blue] {};

	}
}
\foreach \x in {10,14,18,22,26} {
	\foreach \y in {1,5} {            

	\node at (\x,\y) [circle,draw=black,thick,fill=blue] {};

	}
}

\foreach \x in {12,16,20,24} {
	\foreach \y in {1,5} {            

	\node at (\x,\y) [circle,draw=black,thick,fill=red] {};

	}
}

\draw [very thick] (10,5) -- (26,5);
\draw [very thick] (14,5) -- (14,7);
\draw [very thick] (10,7) -- (26,7);
\draw [very thick] (18,5) -- (18,7);

\draw [very thick] (10,1) -- (26,1);
\draw [very thick] (14,1) -- (14,3);
\draw [very thick] (10,3) -- (26,3);
\draw [very thick] (18,1) -- (18,3);

\draw [very thick] (22,3) -- (22,5);
\draw [very thick] (26,3) -- (26,5);

\end{tikzpicture}
}
\caption{}
\end{subfigure}

\caption{(a) `H'-shaped building component of brickwork state. (b) Same with cNOT gates where the control is always $\ket{+}$ and the target always $\ket{0}$. (c) Translation into prime (blue)/dual (red) topologically protected qubits. (d) Prime/dual colouring of the topologically protected brickwork state.}
\label{bricktopo}
\end{figure}
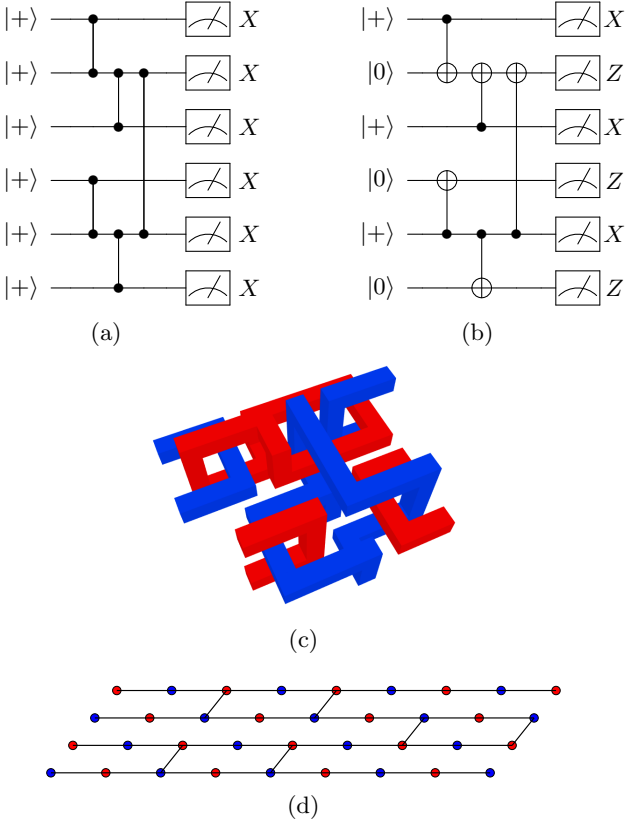

We now adapt Protocol 2a to work with error detection and  attain a better threshold. The main idea is that because the traps are isolated qubits one can look at the syndrome measurements of all trap computations and pick from each computation only the logical trap qubit measurements that are correct individually.

The traps in this case test topologically protected qubits of the graph state that implements the target computation \emph{together with the distillation}. This is because we want to have smaller traps, in terms of number of physical qubits, compared to Protocol 2a where a trap can be as large as a magic state distillation circuit. This is crucial because by employing error detection and retransmission, one needs to resend one logical trap every time at least one syndrome measurement in the topologically protected region around the trap fails (we limit this overhead to $M$ times). 

To avoid the trap computation being distinguishable from the target, we implement the  traps as if all qubits have an injected singular qubit, but the injected singular qubit is prepared in the $\ket{+}$ state and therefore the logical input remains the logical $\ket{+}$. The underlying MBQC blindness hides the qubit that is injected. Each trap computation is now performed $M$ times.

\vspace{0.4cm}

\textbf{Protocol 2b:}

\begin{enumerate}
\item \textbf{Generation of the `topological code-compatible' MBQC pattern:} Verifier selects a random ordering of $2\kappa M +1$ sufficiently large 3D graphs of Fig.~(\ref{3dgraph}), one for the target computation and $2\kappa M$ for the trap computations.
For the target computation round: The Ising sampler MBQC pattern of Fig.~(\ref{brick1}) is translated into a `topological code-compatible' one, i.e. an MBQC pattern where qubits are prepared as $\ket{+_{k \pi/8}}$ states and always measured in the $\{\ket{\pm}\}$ basis. This translation is possible  using again circuits similar to Fig.~(\ref{teleport1}). Notice that this introduces some adaptive $T$ gates that we cannot perform if we want to keep the model instantaneous - this will be accounted for in our supremacy proof. Moreover, topological protection requires the distillation of the magic states and this can also be translated into an MBQC pattern. To avoid adaptivity we fix which magic states we keep for the next level of distillation independently of the syndrome measurements outcomes - this is not a problem because we have `free' postselection on the target computation (register $y$). The final MBQC pattern can be also standardised to the form of a brickwork state so that it can be trapified shown as in Fig.~(\ref{protocol_pic}).

\item \textbf{Trapification:} The target computation is the MBQC pattern generated in the previous step. For the trap round, as shown in Fig.~(\ref{protocol_pic}), we have two types of trap computations by isolating qubits of the brickwork state. This is also `topological code-compatible'. Qubits are prepared in the $\ket{+}$ or $\ket{0}$ state and are measured in the  $\{\ket{\pm}\}$ basis. Notice that in the trap rounds there is no adaptivity. We call this the logical layer of our protocol.

\item \textbf{Topological translation:} As shown in Fig.~(\ref{bricktopo}) one can translate the `topological code-compatible' MBQC pattern into a topology that conforms with the topological code. To avoid leaking any information concerning when magic states or dummy qubits are injected, we inject a physical qubit at every logical qubit. Thus, we use the same topology to inject $\ket{+}$ (which is equivalent to not injecting  anything in the topology of Fig.~(\ref{bricktopo})) or $\ket{0}$ or a magic state when needed.

\item \textbf{Blind implementation of topology:} In order to implement the above topology blindly, so that the prover does not know which physical states we inject, we chose to implement it on MBQC and use the following encryption.

\begin{enumerate}

\item Verifier prepares, one by one, the pre-rotated physical qubits $\ket{+_{\theta}}$,  $\theta \in_R \{0, \frac{\pi}{4},\frac{2\pi}{4}, \ldots,  \frac{7\pi}{4}\}$

\item Verifier sends all the encrypted measurement angles $\delta_i = \theta_i + (-1)^{r'_i}\phi_i + r_i \pi$ for $r_i, r'_i \in_R \{0,1\}$. Parameters $r_i, r'_i$ are classical one-time pads for the measurement outputs.

\item Prover runs the computation, by entangling, measuring all at once and returning the results.
\item Verifier classically detects the errors in the returned syndrome measurements of the trap computations after undoing the $r,r'$ pad. From the set of  the $\kappa M$ logical trap qubits corresponding to each position in the trap graph it selects $\kappa$ correct ones. This is possible, on average, if $M$ is large enough as described at the end of Appendix~\ref{threshh}. This results in the quantity $\Delta_{\kappa}$ in Theorem~\ref{thm2} being averaged over the noise distribution.
\item Verifier accepts the outcome,  which is the logical output string $\boldsymbol{x}$ and  syndrome measurement bit $y$ ($y=0$ no error, $y=1$ error) of the target computation, if all the results of the logical trap computations are correct, otherwise rejects.
\end{enumerate}
\end{enumerate}

\vspace{0.3cm}
The proof of correctness and soundness are identical to Protocol 2a.

\section{Noisy Computational Supremacy}
\label{sec:supr}

Assuming the following conjectures, the quantum computational supremacy theorem (Theorem \ref{thm3}) for the noisy case holds.

\begin{conjecture}[Average-case hardness]
\label{conject2}

For $0 \leq \alpha_1, \beta_1 \leq 1,$  approximating the probability distribution of the Ising sampler by $p^{\mathrm{apx}}(\boldsymbol{x}|y=0)$ up to multiplicative error 
\begin{equation}
| p^{\mathrm{apx}}(\boldsymbol{x}|y=0) - q^{\mathrm{exc}}(\boldsymbol{x}|y=0)|  \leq \alpha_1 q^{\mathrm{exc}}(\boldsymbol{x}|y=0) \nonumber 
\end{equation}
in time $poly(|\boldsymbol{x}|,1/\alpha_1,1/\beta_1)$ is $\#P$-hard for at least a fraction $\beta_1$  of $\boldsymbol{x}$ instances.

\end{conjecture}

\begin{conjecture}[Anti-concentration]
There exist some $0 \leq \alpha_2, \beta_2 \leq 1$, $1/\alpha_2 \in \mathrm{poly}(1/\beta_2)$ such that for all $x$
\begin{eqnarray}
\mathrm{prob}\left( q^{\mathrm{exc}}(\boldsymbol{x}|y=0) \geq \frac{ \alpha_2}{2^{N}} \right) \geq \beta_2
\end{eqnarray}
\label{conject1}
\end{conjecture}

The above encapsulate two properties for the Ising sampler: the worst to average case hardness equivalence for multiplicative approximations and the probability anti-concentration conjecture.

\begin{theorem}[Fault-tolerant hardness]
\label{thm3}
Assume that Conjectures  \ref{conject2} and \ref{conject1} hold. 
Then sampling from the output distribution of the experimental Ising sampler $q^{\mathrm{nsy}}(\boldsymbol{x},y)$ with a classical machine, assuming a $(\varepsilon',\varepsilon)$-sound verification scheme
 (Def.~\ref{def_ver}/Def.~\ref{def_ver_post})  accepts with 
\begin{equation}
\label{eq:epsi}
\varepsilon \leq \frac{(\beta_1+\beta_2 -1-2^{-N})\alpha_1 \alpha_2}{2},
\end{equation}
implies, with confidence $\varepsilon'$, a collapse in the polynomial hierarchy to the third level.
\end{theorem}

In order to have a scheme with positive  soundness parameter $\varepsilon$, we need our conjectures to satisfy  $\beta_1 + \beta_2 -2^{-N} \geq 1$.

The proof uses Stockmeyer's theorem \cite{stockmeyer1985approximation}, which is based on a hypothetical machine and predicts, if classical sampling is possible, the collapse of the polynomial hierarchy to the third level. The collapse of the infinite polynomial hierarchy at any level is considered a highly unlikely event in computational complexity theory.

\subsection{Proof}

Compared the the FT hardness proof of \cite{fujii2016noise}, our proof is for the more general case of additive as opposed to multiplicative approximation, thus answering an open question of that paper.

 We follow a similar line of reasoning as the original translationally invariant Ising sampler~\cite{gao_quantum_2016} - proof by contradiction. The main difference is that we use probabilities conditioned on null syndromes. Other differences include adding explanation of intermediate steps, a discussion about obfuscation and breaking the original single Conjecture into two separate ones: one for anti-concentration and one for average case hardness.
 
  We also follow a line similar to an earlier result \cite{d71a5fed54df414db082ed64ed5c9ef7}. Compared to that result our proof does not assign specific numbers to the parameters of the conjectures, but states them in a parametrised fashion.
  
The following proof holds for a verification scheme according to Def.~\ref{def_ver_post}. In the case of a verification scheme of Def.~\ref{def_ver} the same proof holds, replacing var$^{\text{Post}}$ with var and setting $q^{\text{nsy}}(y=0)=1$.

\begin{proof}[Proof of Theorem \ref{thm3}]

If we can classically sample from $q^{\text{nsy}}(\boldsymbol{x},y)$ (which means that our quantum computer could be a classical impostor), then estimating the probabilities of the distribution
 with exponential accuracy
 is in $\#P$: We can construct a polynomial time non-deterministic Turing machine that uses the sampler as an oracle that accepts when a specific string is sampled, so that the probability of that event could be estimated, if we could count the accepting branches.  We could also estimate the marginal probabilities $q^{\text{nsy}}(y)$ in such a manner. Notice that we could not apply the same argument for the quantum sampler since we cannot extract its randomness as input to build the oracle.
   From Stockmeyer's theorem~\cite{stockmeyer1985approximation}, there exists an $FBPP^{NP}$ machine that can compute explicitly the values  $p^{\text{apx}}(\boldsymbol{x},y)$, such that for every $\boldsymbol{x},y$
\begin{equation}
| p^{\text{apx}}(\boldsymbol{x},y) - q^{\text{nsy}}({\boldsymbol{x}},y) | \leq  \frac{q^{\text{nsy}}({\boldsymbol{x}},y)}{\text{poly}(N)}.
\end{equation}

The same can be applied in calculating the marginals. Thus a $FBPP^{NP}$ machine can calculate $q^{\text{nsy}}(y=0)$, the probability of accepting the syndrome measurements, with accuracy of the same scaling as the joint probability.

Using the fact that $q^{\text{nsy}}(y=0)$ is non-zero (it is lower bounded by $(1- \epsilon)^N$, so one can get a non-zero estimate in $\#P$ and approximate it using Stockmeyer) it is easy to prove that for conditional probabilities,
\begin{equation}
| p^{\text{apx}}(\boldsymbol{x}|y=0) - q^{\text{nsy}}({\boldsymbol{x}}|y=0) | \leq  \frac{q^{\text{nsy}}({\boldsymbol{x}}|y=0)}{\text{poly}(N)}.
\end{equation}

Applying the triangle inequality, for every $\boldsymbol{x}$ the distance between the values
$p^{\text{apx}}(\boldsymbol{x}|y=0)$ and the exact conditional probability $q^{\text{exc}}(\boldsymbol{x}|y=0)$ of the Ising sampler is
\begin{eqnarray}
 | p^{\text{apx}}(\boldsymbol{x}|y=0) - q^{\text{exc}}(\boldsymbol{x}|y=0)|  \nonumber \\
 \leq | p^{\text{apx}}(\boldsymbol{x}|y=0) - q^{\text{nsy}}(\boldsymbol{x}|y=0)| \nonumber \\  + | q^{\text{nsy}}(\boldsymbol{x}|y=0) - q^{\text{exc}}(\boldsymbol{x}|y=0)|  \\
 \leq \frac{q^{\text{nsy}}(\boldsymbol{x}|y=0)}{\text{poly}(N)} +  | q^{\text{nsy}}(\boldsymbol{x}|y=0) \nonumber \\- q^{\text{exc}}(\boldsymbol{x}|y=0)|.
\label{eq:threeterms}
\end{eqnarray}
Assuming an $(\varepsilon',\varepsilon)$-sound verification scheme has accepted, it follows that  $ | q^{\text{nsy}}(\boldsymbol{x}|y=0) - q^{\text{exc}}(\boldsymbol{x}|y=0)| \leq 2\varepsilon$,  with confidence $\varepsilon'$. 

Obfuscation of the probability estimated in this model comes by construction. We can pick a computational branch $(x_1, \ldots, x_{m(n-1)})$ and final layer output string $(x_{n}, \ldots, x_{m n})$ to estimate at random, without revealing any information to the sampler. This is possible because the uniform distribution over branches  (see Eq.~(\ref{eqnequi}) in Section~\ref{sec:correct}) is created within a fixed instance of the Ising sampler, with no extra input provided to the sampler. 
 The expectation of $\text{var}^{\text{Post}}$ over the uniform distribution on $\boldsymbol{x}$ is $\leq \frac{2\varepsilon}{2^{mn}}$, where $m,n$ are the dimensions of the logical `extended' brickwork state and $N=mn$.

Markov inequality relates the probability of a random variable exceeding a certain value with its expectation. For a random variable $X$ and $\gamma >0,$
\begin{equation}
\text{prob}(X \geq \gamma  ) \leq \frac{E(X)}{\gamma}.
\end{equation}

Applying the Markov inequality to the second term in Eqn.~(\ref{eq:threeterms})
\begin{equation}
\text{prob}( | q^{\text{nsy}}(\boldsymbol{x}|y=0) - q^{\text{exc}}(\boldsymbol{x}|y=0)| \geq \gamma ) \leq  \frac{2\varepsilon}{2^{N} \gamma}. 
\end{equation}
By changing variables
\begin{eqnarray}
\text{prob}\left( | q^{\text{nsy}}(\boldsymbol{x}|y=0) - q^{\text{exc}}(\boldsymbol{x}|y=0) | \right. \nonumber \\  \left. \geq  \frac{2\varepsilon}{2^{N} \gamma} \right)  \leq \gamma 
\end{eqnarray}

or

\begin{eqnarray}
\text{prob}\left( | q^{\text{nsy}}(\boldsymbol{x}|y=0) - q^{\text{exc}}(\boldsymbol{x}|y=0) | \right. \nonumber \\ \left. \leq  \frac{2\varepsilon}{2^{N} \gamma} \right)  \geq 1 - \gamma. 
\end{eqnarray}

Condensing this with Eqn.~(\ref{eq:threeterms}),

\begin{eqnarray}
\text{prob} \Bigg( | p^{\text{apx}}(\boldsymbol{x}|y=0) - q^{\text{exc}}(\boldsymbol{x}|y=0)|  \leq  \nonumber \\ \frac{q^{\text{exc}}(\boldsymbol{x}|y=0)}{\text{poly}(N)} 
 \left. + \frac{2\varepsilon (1+o(1))}{2^{N} \gamma} \right)  \nonumber \\
\geq 1- \gamma,   
\end{eqnarray}

Thus, for more than $1-\gamma$ fraction of instances of $\boldsymbol{x}$
\begin{eqnarray}
| p^{\text{apx}}(\boldsymbol{x}|y=0) - q^{\text{exc}}(\boldsymbol{x}|y=0)| \nonumber \\ \leq \frac{q^{\text{exc}}(\boldsymbol{x}|y=0)}{\text{poly}(N)}  + \frac{2\varepsilon (1+o(1))}{2^{N} \gamma}, \label{fractionof} 
\end{eqnarray}

which means that strongly simulating, i.e. calculating the probabilities, of the Ising distribution with the above mixture of additive and multiplicative accuracy for more than $1-\gamma$  fraction of instances of $\boldsymbol{x}$ is in the third level of the polynomial hierarchy.

We use the two conjectures to continue our proof. From Eqn. (\ref{fractionof}) and Conjecture~\ref{conject1}, setting $\varepsilon_1 \in \frac{2\varepsilon (1+o(1))}{\gamma \alpha_2}$, there must be at least $\beta_2 - \gamma$ fraction of instances of $\boldsymbol{x}$ (we assume $\gamma < \beta_2$)
 such that
\begin{eqnarray}
| p^{\text{apx}}(\boldsymbol{x}|y=0) - q^{\text{exc}}(\boldsymbol{x}|y=0)|    \nonumber \\
\leq \frac{q^{\text{exc}}(\boldsymbol{x}|y=0)}{\text{poly}(N)} + \varepsilon_1 q^{\text{exc}}(\boldsymbol{x}|y=0) \\
 \leq  (o(1) + \varepsilon_1) q^{\text{exc}}(\boldsymbol{x}|y=0).
\end{eqnarray}

Let Conjecture~\ref{conject2} hold with $\beta_1 \geq 1 - (\beta_2 - \gamma) + 2^{-N}$  
 and $\alpha_1 \in o(1) + \varepsilon_1$.  These imply for the soundness parameter
\begin{equation}
\varepsilon \leq \frac{(\beta_1+\beta_2 -1-2^{-N})\alpha_2 \alpha_1}{2} ,
\end{equation}
which is positive for $\beta_1 + \beta_2 -2^{-N} \geq 1$.

Then, there exists at least one instance for which the  multiplicative approximation the $FBPP^{NP}$ Stockmeyer machine have calculated is $\#P$-hard.

Then, $PH \subseteq P^{\# P} \subseteq P^{FBPP^{NP}}\subseteq \Sigma_3^{P}$, where the first inclusion is given by Toda's theorem, and the polynomial hierarchy collapses to the third level, an event expected to be highly unlikely.

\end{proof}

Our average case conjecture is implicitly contained in  a stronger conjecture that includes anti-concetration in~\cite{gao_quantum_2016}. Notice that our average case conjecture, however, applies to a slightly different distribution. The difference is that, in our case, we can have extra rotations on some of the measurement angles, as a consequence of not applying the correction gates conditioned on the measurement outcome of the teleportation step of the FT gates (Fig.~(\ref{teleport1}) $(ii)$). This issue will affect the implementation of the measurements with non-zero angles in the extended brickwork state (Figure \ref{brick1} $(ii)$). Assuming magic $\ket{+_{\pi/8}}$ states are used for the implementation of the $\pi/8$ rotations ($\sqrt{T}$ gates), the byproduct is a $k_0 \pi/4$ rotation on the measurements of the brickwork state vertices (Fig.~(\ref{brick1}) $(i)$), for some $k_0$ which depends on the measurement outcomes of the magic state teleportation steps. The original argument made in \cite{gao_quantum_2016} to support their average case conjecture  relies on the fact that from random measurement outcomes of the vertices of the extended brickwork state, a uniform rotation over the $8$ different $\{k \pi/4\}$ angles is produced on the brickwork state. The argument is that this corresponds to random circuits which will likely produce highly entangled states (see also p.~$7$ in Supplemental Material of Ref.~\cite{gao_quantum_2016}). In our case it is the same, because the extra rotations cannot change the uniformity of these angles. Thus, the computation applied is based on a random brickwork MBQC pattern and connections to the random circuit model, such as in \cite{boixo2018characterizing}, can be made. In the latter paper, the average hardness of sampling from a random circuit is supported by drawing connections to quantum chaos and some numerical evidence. In another recent result~(\cite{bouland2018quantum}), average-case hardness has been proven for random circuits for the exact case.

\section{Discussions}
\label{sec:Disc}

Quantum computational supremacy demonstration is believed to be easier than universal quantum computing since it may not have to fulfil atleast one of DiVincenzo's criteria. 
Our work shows that fault-tolerant verifiable quantum supremacy is quantitatively easier than fault-tolerant universal quantum computation in terms of thresholds. 
This relies on combining the notion of post-selected thresholds with trap-based verification which allows error-detection-based fault tolerance to combat noise. 
Such a combination is not known to exist for other quantum verification methods.
In the trap-based verification schemes we use, it is the isolatable nature of the traps that enables error-detection-based fault tolerance.

The techniques developed here have a wide range of applicability. We apply it to the Ising sampler as a specific example of a model for quantum supremacy. For example, they could be applied in implementations of the Boson Sampling model~\cite{aaronson10boson_sampling} in a fault-tolerant quantum computer based on qubits~\cite{peropadre2015spin}. Our methods should also apply to the random circuit IQP model~\cite{PhysRevLett.117.080501} which, in the `graph program' implementation~\cite{shepherd09iqp} requires a smaller than the Ising sampler, but non-planar, resource state. Finally, it can be applied to recently studied quantum supremacy architectures on low-periodicity planar lattices~\cite{bermejo2017architectures}. The only requirement for our isolated trap computation technique is that the underlying graph state is bipartite.  Thus, it can even be used to simplify the original verification protocol \cite{fitzsimons12vubqc} for a universal quantum computer, in the case it runs a classical output problem, and use our technique to implement it in a fault-tolerant way.

Trap-based techniques require blindness, which is not believed possible with a classical verifier~\cite{morimae2014impossibility, aaronson2017implausibility}. Even verification protocols that do not require blindness, such as~\cite{aharonov10abe_published}, still need some level of quantum encryption. This is true for any protocol based on quantum authentication schemes~\cite{barnum02qas}, made possible by a quantum verifier. Verification protocols with classical verifiers exist~\cite{shepherd09iqp,reichardt12ruv_nature,mahadev2018classical}, but require extra assumptions such as additional computational hardness conjectures or non-communicating provers respectively. For general reviews of blind and verifiable protocols see Refs.~\cite{fitzsimons2017private,gheorghiu2018verification}.

Our work is one of the first on fault-tolerant verification, which was known to be a challenging open question. Another recent progress \cite{fujii2016verifiable} presents a fault-tolerant verification technique for universal MBQC that requires the verifier to perform measurements, as opposed to preparations as in our scheme.  Our scheme is complementary to contemperaneous work on composable verification of IQP, which is a classical hypothesis test with the verifier preparing perfect stabilizer states and the prover using a non-planar graph~\cite{DanielMills2017}. A formal proof of composability of our protocol is a desirable next step and may be developed using the  methods given in \cite{dunjko14composable}.

A direction for future investigation should be the potential of other known fault-tolerant quantum codes in providing improved post-selected thresholds, as well as the search for quantum codes for non-universal models. Another direction should be the study of known quantum supremacy models for which a verifiable version using the same restricted physical assumptions as the original exists \cite{kapourniotis2014blindness}, as well as the development of such new non-universal models.
More technically, an open problem for our verification scheme is to find a graph state with local rotations being only multiples of $\pi/4$ and still generate random universal logical measurement angles as in the existing scheme. This will make the fault-tolerant version easier because it will be based on more standard magic states. Also, other universal constructions with $xy$-plane measurements can also be considered~\cite{mantri2016universality,PhysRevA.97.022333}.

\section*{Acknowledgements}

We thank Zhang Jiang, Ashley Montanaro, Michael Bremner for early discussions and Elham Kashefi, Petros Wallden and Andru Gheorghiu for discussions on the trap-based verification technique. We thank the anonymous referees for valuable feedback that helped improve the manuscript. We thank Dominic Branford and Samuele Ferracin for discussions and the former for also helping with the figures. This work was supported, in part, by the UK EPSRC (EP/K04057X/2), and the UK National Quantum Technologies Programme (EP/M013243/1).

\bibliographystyle{plainnat}
\bibliography{VerificationReferences}

\onecolumn\newpage
\appendix

\section{\label{app:bridge}Bridge and break operators and the Ising Sampler}

To understand the procedure of carving a specific graph out of a square lattice state by using only $xy$-plane measurements, we explain two types of operations originally introduced in \cite{fitzsimons12vubqc, hein2004multiparty}.

The first is \emph{break} operators. Let $i$ be a vertex we want to remove from the original graph, together with the connection to its neighbours. One can achieve this by performing a Pauli $Z$ measurement on the qubit that corresponds to $i$ and discard the outcome. However, we do not have the power to perform measurements out of the $xy$-plane in our protocol, otherwise we risk revealing the position of the traps by asking the prover to change the basis. Pauli $Z$-measurements can be simulated by preparing  \emph{dummy} qubits in the $\ket{0}$ state. Then,the $cZ$ gate applied by the prover has no effect in entangling it with its neighbours. The measurement can have any arbitrary rotation since the qubit is isolated and does not participate in the computation. In order to keep the whole procedure blind we instead prepare the qubit in state $ \ket{d_i}$ for $d_i$ chosen independently and uniformly at random from $\{0,1\}$. This ensures  that the state is identical to the maximally mixed state, as is the case for the other qubits in a blind protocol. Also, we apply on each of its neighbours Pauli operation $Z^{d_i}$, before sending them to the prover, so that we cancel the effect that the prover's entangling will have on that neighbour.

The second  is called \emph{bridge} operators. Let $i$ be a vertex of degree $2$ that we want to remove  in the original graph and join its neighbours by an new edge. To achieve this we apply a Pauli $Y$ measurement (measurement angle $\pi/2$) on the qubit that corresponds to vertex $i$ and add $\pi/2$ on the angles of each the two neighbours. Conditioned on this measurement giving $0$ the resulting graph, when we trace out the measured qubit, is the desired one. If the measurement outcome is $1$ then, in order to get the correct graph, we need to apply a $Z$ correction on the neighbours. Since our Ising sampler model is nonadaptive and all our measurements are in the $xy$ plane we can achieve this by flipping the measurement outcomes of the corresponding qubits. Notice that this is not an issue in the trap rounds that we explain in Section~\ref{NFTverif} since there are no bridge operations in this case. 

\section{\label{app:Thm1}Proof of soundness in Theorem~\ref{thm1}}

\begin{proof}
Let $U_P(\boldsymbol{r},\boldsymbol{d})$ denote the correct unitary operation of the protocol. It includes everything preceding $U_B$ in Fig.~(\ref{completeness3}), and we have only included in the arguments the random parameters that will be averaged over later. The vector $\boldsymbol{r}$ contains as elements bits $r_i,r'_i$, where $i$ ranges from $1$ to $(2\kappa+1)N$ ($2^{2(2\kappa+1)N}$ possible values)  and the vector $\boldsymbol{d}$ contains as elements bits $d_i$, where $i$ ranges again from $1$ to $(2\kappa+1)N$ (for the non-dummy qubits we assume fixed $d_i=0$, thus $2^{\kappa N}$ possible values). The rest of the  random parameters are the vector $\boldsymbol{\theta}$ which contains elements $\theta_i \in \{0, \frac{\pi}{8},\frac{2\pi}{8}, \ldots,  \frac{15\pi}{8}\}$ for $1 \leq i \leq (2\kappa+1)N$ ($16^{(2\kappa+1)N}$ possible values) and the vector $\boldsymbol{t}$ which contains the indices of the positions of the traps in all $2\kappa$ trap systems and takes all ${{2 \kappa +1}\choose{\kappa}} (\kappa+1)$ possible values allowed by the construction. The distributions over all the possible values of the above random parameters are uniform.

In the honest case, after $U_P(\boldsymbol{r},\boldsymbol{d})$ is applied, the state of the trap system becomes $\rho_t= \bigotimes_{i \in \boldsymbol{t}} Z^{r_i} \ket{+}_i \bra{+}_i Z^{r_i}$, where the index $i$ takes values from the elements of $\boldsymbol{t}$ that represent the positions of the traps. In the dishonest case (again based on Fig.~(\ref{completeness3})), tracing out the prover's private system, the deviation $U_B$ becomes an arbitrary CPTP map denoted by $\mathcal{E}$. The  probability of getting all zeros of the trap system $\rho'_t$, right after the measurements, can be written as

\begin{eqnarray}
F_t^2 & \equiv & \sum_{\boldsymbol{r},\boldsymbol{\theta},\boldsymbol{t},\boldsymbol{d}} p(\boldsymbol{r},\boldsymbol{\theta},\boldsymbol{t},\boldsymbol{d}) \text{Tr} \left( \bigotimes_{i \in \boldsymbol{t}}  Z^{r_i} \ket{+}_i \bra{+}_i Z^{r_i} \rho'_t \right )  
\end{eqnarray}

\begin{eqnarray}
&=&\sum_{\boldsymbol{r},\boldsymbol{\theta},\boldsymbol{t},\boldsymbol{d}, \boldsymbol{b}} p \text{Tr} \left( \bigotimes_{i \in \boldsymbol{t}}  Z^{r_i} \ket{+}_i \bra{+}_i Z^{r_i}  \text{Tr}_{\{i : i \notin \boldsymbol{t}\}} \left( \bigotimes_i Z^{b_i} \ket{+}_i\bra{+}_i Z^{b_i} \mathcal{E} \left(U_P(\boldsymbol{r},\boldsymbol{d}) \bigotimes_{i \notin \boldsymbol{m(t)}} \ket{+}_i \bra{+}_i \right. \right. \right. \nonumber \\ & & \bigotimes_{i \in \boldsymbol{m(t)}} \ket{0}_i \bra{0}_i  U_P(\boldsymbol{r},\boldsymbol{d})^{\dagger} 
 \left. \left. \left. \bigotimes_i \ket{\delta_i(\theta_i,r_i)} \bra{\delta_i(\theta_i,r_i)}\right) \bigotimes_i Z^{b_i} \ket{+}_i\bra{+}_i Z^{b_i}\right)\right) 
\end{eqnarray}

where the inner trace in the formula is taken over all the systems except the trap system. The vector $\boldsymbol{b}$ has been introduced, where elements $b_i$ are bits which correspond to the results of measurements of bits $i$ for $1 \leq i \leq (2\kappa+1)N$. The probability $p$ comes from the uniform distribution over all possible values of the random parameters and is therefore $1/(2^{2(2\kappa+1)N} 2^{\kappa N}16^{(2\kappa+1)N}{{2 \kappa +1}\choose{\kappa}} (\kappa+1))$. Also, $\boldsymbol{m(t)}$ are the positions of the dummy qubits for a choice of trap positions $\boldsymbol{t}$.

Summing over $\theta$'s creates the maximally mixed state for the $\delta$'s and summing over the $r$'s and $d$'s of  the computational system and the dummy system creates the maximally mixed state for those systems. This is because just before the application of deviation operator these systems are not entangled with the trap system and at the same time a quantum one-time-pad is applied on them. We can therefore trace them out and update the CPTP map $\mathcal{E}$ to a new CPTP map $\mathcal{E}'$ that applies on the remaining system (of dimension $2^{N'}$) and does not depend on the secret parameters.

The CPTP map $\mathcal{E}'$ can be written as a Kraus decomposition, where the Kraus operators $\{E_u\}$ obey $\sum_u E_u E_u^{\dagger} = I_{2^{N'}}$, where $I_{2^{N'}}$ is the identity on a $2^{N'}$ dimension system. Each Kraus operator can be further decomposed into the Pauli basis as $E_u = \sum_k a_{u,k} P_k$, where $\{P_k\}$ are all generalized elements of the Pauli basis applying on a $2^{N'}$ dimension system and $\{a_{u,k}\}$ are complex coefficients. Also, we remind the reader that the $\phi$ parameters of the trap qubits are all zero and therefore the remaining honest operation consists only of the rotations by the $r$ parameters.

\begin{eqnarray}
F_t^2 & = & \sum_{\boldsymbol{r_t},\boldsymbol{t},\boldsymbol{b_t}}p(\boldsymbol{r_t},\boldsymbol{t})\text{Tr}\bigg( \sum_u  \sum_{k=1}^{4^{N'}} \sum_{l=1}^{4^{N'}} a_{u,k} a_{u,l}^*  \bigotimes_{i \in \boldsymbol{t}} Z^{r_i} \ket{+}_i \bra{+}_i Z^{r_i}  \bigotimes_{i \in \boldsymbol{t}} Z^{b_i} \ket{+}\bra{+} Z^{b_i} P_{k | i} Z^{r_i} \ket{+}_i \bra{+}_i Z^{r_i} P_{l | i} \nonumber \\
& & Z^{b_i} \ket{+}\bra{+} Z^{b_i} \bigg),
\end{eqnarray}
where $P_{k | i}$ denotes the Pauli operator that applies on qubit with index $i$ from the generalized Pauli basis operator $P_k$. Because of the state of the system, in particular the fact that $X\ket{+}=\ket{+}$, we can add `free' Pauli $X$ operators randomized by new parameters $r'$ taken uniform at random.

\begin{eqnarray}
F_t^2 &=& \sum_{\boldsymbol{r_t}, \boldsymbol{r'_t},\boldsymbol{t},\boldsymbol{b_t}}p(\boldsymbol{r_t}, \boldsymbol{r'_t},\boldsymbol{t})\text{Tr} \bigg(  \sum_{u,k,l} a_{u,k} a_{u,l}^* \bigotimes_{i \in \boldsymbol{t}} Z^{r_i}  \ket{+}_i \bra{+}_i Z^{r_i}  \bigotimes_{i \in \boldsymbol{t}}  Z^{b_i} \ket{+}\bra{+}  X^{r'_i} Z^{b_i} P_{k | i} Z^{r_i} X^{r'_i} \ket{+}_i \bra{+}_i \nonumber \\
& & X^{r'_i} Z^{r_i} P_{l | i}  Z^{b_i} X^{r'_i} \ket{+}\bra{+} Z^{b_i} \bigg) 
\end{eqnarray}

By changing variables $b_i' = b_i + r_i$ and applying the cyclic property of the trace to move $Z^{r_i} X^{r'_i}$ around
\begin{eqnarray}
F_t^2 &=& \sum_{\boldsymbol{r_t}, \boldsymbol{r'_t},\boldsymbol{t},\boldsymbol{b'_t}}p(\boldsymbol{r_t}, \boldsymbol{r'_t},\boldsymbol{t})\text{Tr}\bigg(  \sum_{u,k,l} a_{u,k} a_{u,l}^* \bigotimes_{i \in \boldsymbol{t}}   \ket{+}_i \bra{+}_i   \bigotimes_{i \in \boldsymbol{t}}  Z^{b'_i} \ket{+}\bra{+}  X^{r'_i} Z^{b'_i+r_i} P_{k | i} Z^{r_i} X^{r'_i} \ket{+}_i \bra{+}_i           \nonumber \\
& & X^{r'_i} Z^{r_i} P_{l | i}  Z^{b'_i+r_i} X^{r'_i}  \ket{+}\bra{+} Z^{b'_i} \bigg)
\end{eqnarray}

Applying the Pauli twirl lemma~\cite{dankert09operator_twirl}, proven in Appendix~\ref{app:lemma}, by averaging over $\boldsymbol{r_t}, \boldsymbol{r'_t},$ we get

\begin{eqnarray}
F_t^2 &=& \sum_{\boldsymbol{t},\boldsymbol{b'_t}} p(\boldsymbol{t})   \sum_{u,k} |a_{u,k}|^2  \prod_{i \in \boldsymbol{t}}  |\bra{+}_i  Z^{b'_i} \ket{+}\bra{+} Z^{b'_i}   P_{k |i }  \ket{+}_i |^2 \nonumber \\
&=& \sum_{\boldsymbol{t}} p(\boldsymbol{t})   \sum_{k} |\alpha_{k}|^2  \prod_{i \in \boldsymbol{t}}  |\bra{+}_i    P_{k |i }  \ket{+}_i |^2 ,
\label{eq1_}
\end{eqnarray}

where $ |\alpha_{k}|^2 = \sum_{u} |a_{u,k}|^2 $ and $\sum_{k} |\alpha_{k}|^2 = 1$ from the unital property of the attack.

A similar analysis is applied to calculate the average fidelity $F_c=F(\rho_c, \rho'_c)$ of the computational state after the measurements. In the honest case the computational state $\rho_c$ just before the measurement will be disentangled from the rest of the system: $ \bigotimes_{i \in \boldsymbol{c}} Z^{r_i} X^{r'_i}R_z(\phi_i)X^{r'_i} \ket{G} \bra{G} \bigotimes_{j \in \boldsymbol{c}} X^{r'_j} R_z(-\phi_j) X^{r'_j} Z^{r_j}$, where $\boldsymbol{c(t)}$ are the positions of the computational qubits for a choice of trap positions $\boldsymbol{t}$ and $\ket{G} \bra{G}$ is the computational graph state. The latter can be expressed as $E_G \bigotimes_{i \in \boldsymbol{c}} \ket{+}_i \bra{+}_i E_G^{\dagger}$, where $E_G$ denotes all entangling operators $cZ$ that apply on a graph $G$. In the dishonest case, for an attack $\mathcal{E}$ the fidelity $\bar{F}_c$ is

\begin{eqnarray}
F_c &\equiv& \sum_{\boldsymbol{r},\boldsymbol{r}',\boldsymbol{\theta},\boldsymbol{t}, \boldsymbol{d}} p(\boldsymbol{r},\boldsymbol{r}',\boldsymbol{\theta},\boldsymbol{t}, \boldsymbol{d}) F(\rho_c, \rho'_c)  \\  
& = & \sum_{\boldsymbol{r},\boldsymbol{r}',\boldsymbol{\theta},\boldsymbol{t}, \boldsymbol{d}} p(\boldsymbol{r},\boldsymbol{r}',\boldsymbol{\theta},\boldsymbol{t}, \boldsymbol{d})  \text{Tr}\left( (\rho_c \rho'_c)^{1/2} \right)   \\
 &\geq&    \text{Tr}\left( \left (\sum_{\boldsymbol{r},\boldsymbol{r}',\boldsymbol{\theta},\boldsymbol{t}, \boldsymbol{d}} p(\boldsymbol{r},\boldsymbol{r}',\boldsymbol{\theta},\boldsymbol{t}, \boldsymbol{d}) \rho_c \rho'_c \right)^{1/2} \right)    \\
&=& \text{Tr}\left(\left(\sum_{\boldsymbol{r},\boldsymbol{r}', \boldsymbol{\theta},\boldsymbol{t}, \boldsymbol{d},\boldsymbol{b}} p \bigotimes_{i \in \boldsymbol{c}} Z^{r_i} X^{r'_i} R_z(\phi_i) X^{r'_i} \ket{G} \bra{G}  \bigotimes_{j \in \boldsymbol{c}} X^{r'_j} R_z(-\phi_j) X^{r'_j} Z^{r_j} 
\text{Tr}_{\{i : i \notin \boldsymbol{c}\}} \left( \bigotimes_i Z^{b_i} \ket{+}_i \right. \right. \right.\nonumber \\
&& \bra{+}_i   Z^{b_i} \mathcal{E} ( U_P(\boldsymbol{r},\boldsymbol{r}',\boldsymbol{d})  \bigotimes_{i \notin \boldsymbol{m(t)}} \ket{+}_i \bra{+}_i \bigotimes_{i \in \boldsymbol{m(t)}} \ket{0}_i \bra{0}_i  U_P(\boldsymbol{r},\boldsymbol{r}',\boldsymbol{d})^{\dagger} 
  \bigotimes_i \ket{\delta_i(\theta_i,r_i,r'_i)} \bra{\delta_i(\theta_i,r_i,r'_i)})  \nonumber \\
&&  \left. \left. \left. \bigotimes_i Z^{b_i} \ket{+}_i\bra{+}_i Z^{b_i}\right) \right)^{1/2} \right).
\end{eqnarray}

Summing over the $\theta$'s of the $\delta$'s and the $r$'s and the $d$'s of the trap and the dummy system creates the maximally mixed system for these systems which can be traced over. Expressing the attack on the remaining system (of dimension $2^{N'}$) using the Kraus decomposition with each Kraus element decomposed in the Pauli basis we get as before

\begin{eqnarray} 
F_c  \geq   \text{Tr}\bigg( \bigg( \sum_{\boldsymbol{t},\boldsymbol{r_{c(t)}},\boldsymbol{r'_{c(t)}},\boldsymbol{b_{c(t)}}}  p \sum_{u} \sum_{k=1}^{4^{N'}} \sum_{l=1}^{4^{N'}}  a_{u,k} a^{*}_{u,l} \bigotimes_{i \in \boldsymbol{c}} Z^{r_i} X^{r'_i} R_z(\phi_i) X^{r'_i} \ket{G} \bra{G}  
 \bigotimes_{j \in \boldsymbol{c}}X^{r'_j} R_z(-\phi_j) 
 X^{r'_j} Z^{r_j} \nonumber \\ 
  \bigotimes_{i \in \boldsymbol{c}} Z^{b_i} \ket{+}\bra{+} Z^{b_i} 
  P_{k | i} Z^{r_i} X^{r'_i} R_z(\phi_i) X^{r'_i} \ket{G} \bra{G} \bigotimes_{j \in \boldsymbol{c}} X^{r'_j} R_z(-\phi_j) X^{r'_j} Z^{r_j} P_{l | j} Z^{b_i} \ket{+}\bra{+} Z^{b_i}\bigg)^{1/2}\bigg).
\end{eqnarray}

By changing variables $b' = b + r$ and applying the cyclic property of the trace
 \begin{eqnarray}
F_c  \geq  \text{Tr}\bigg( \bigg( \sum_{\boldsymbol{t},\boldsymbol{r_{c(t)}},\boldsymbol{r'_{c(t)}},\boldsymbol{b'_{c(t)}}}  p \sum_u \sum_{k=1}^{4^{N'}} \sum_{l=1}^{4^{N'}}  a_{u,k} a^{*}_{u,l} \bigotimes_{i \in \boldsymbol{c}}  X^{r'_i} R_z(\phi_i) X^{r'_i} \ket{G} \bra{G}  
 \bigotimes_{j \in \boldsymbol{c}}X^{r'_j} R_z(-\phi_j) 
 X^{r'_j}  \bigotimes_{i \in \boldsymbol{c}} Z^{b'_i}  
    \nonumber \\
 \ket{+} \bra{+} Z^{b'_i + r_i} P_{k | i} Z^{r_i} X^{r'_i} R_z(\phi_i) X^{r'_i} \ket{G} \bra{G} \bigotimes_{j \in \boldsymbol{c}} X^{r'_j} R_z(-\phi_j) X^{r'_j} Z^{r_j} P_{l | j} Z^{b'_i+r_i} \ket{+}\bra{+} Z^{b'_i}\bigg)^{1/2} \bigg).
\end{eqnarray}

Using Corollary \ref{twirl2} in Appendix~\ref{app:lemma} and the cyclic property of the trace and sum over $\boldsymbol{r_{c(t)}}$, we can eliminate all Pauli $X$ operators of the attack that differ in the two sides (we denote this by replacing the summation over $l$ with a summation over $l_x$ where the element $P_l$ agrees with $P_k$ on all the Pauli $X$ components). We also use the cyclic property of the trace to get
 \begin{eqnarray}
F_c & \geq & \text{Tr}\bigg( \bigg( \sum_{\boldsymbol{t},\boldsymbol{r_{c(t)}},\boldsymbol{r'_{c(t)}},\boldsymbol{b'_{c(t)}}}  p \sum_u \sum_{k} \sum_{l_x}  a_{u,k} a^{*}_{u,l}  \bra{G}  
 \bigotimes_{j \in \boldsymbol{c}}X^{r'_j} R_z(-\phi_j) 
 X^{r'_j}  
  \bigotimes_{i \in \boldsymbol{c}} Z^{b'_i} \ket{+}\bra{+} Z^{b'_i} P_{k | i}  X^{r'_i} R_z(\phi_i) \nonumber \\
 & &  X^{r'_i} \ket{G} \bra{G} \bigotimes_{j \in \boldsymbol{c}} X^{r'_j} R_z(-\phi_j) X^{r'_j}  P_{l | j} Z^{b'_i} \ket{+}\bra{+} Z^{b'_i} \bigotimes_{i \in \boldsymbol{c}}  X^{r'_i} R_z(\phi_i) X^{r'_i} \ket{G}\bigg)^{1/2} \bigg). 
\end{eqnarray}

Then, the $X^{r'_i}$ ($X^{r'_j}$) operators that are next to $\ket{G}$ ($\bra{G}$) can be rewritten as Pauli $Z$ operators on their neighbours, which then commute with $z$ rotations and the attack (since the Pauli $X$ operators of the attack are the same from both sides) and with the projectors $Z^{b'_i} \ket{+}\bra{+} Z^{b'_i}$, by changing variable $b''_i = b'_i+r'_i$, and cancel each other. The $X^{r'_i}$ operators that are next to projectors $Z^{b'_i} \ket{+}\bra{+} Z^{b'_i}$ commute with them trivially
 \begin{eqnarray}
F_c & \geq & \text{Tr}\bigg( \bigg( \sum_{\boldsymbol{t},\boldsymbol{r'_{c(t)}},\boldsymbol{b''_{c(t)}}}  p \sum_u \sum_{k} \sum_{l_x}  a_{u,k} a^{*}_{u,l}  \bra{G}  
 \bigotimes_{j \in \boldsymbol{c}} R_z(-\phi_j)    
  \bigotimes_{i \in \boldsymbol{c}} Z^{b''_i} \ket{+}\bra{+} Z^{b''_i} \nonumber \\
 & & X^{r'_j} P_{k | i} X^{r'_i} R_z(\phi_i)  \ket{G} \bra{G} \bigotimes_{j \in \boldsymbol{c}}  R_z(-\phi_j) X^{r'_j}  P_{l | j} X^{r'_i} Z^{b''_i} \ket{+}\bra{+} Z^{b''_i} \bigotimes_{i \in \boldsymbol{c}}   R_z(\phi_i)  \ket{G}\bigg)^{1/2} \bigg).
\end{eqnarray}

 Then we can use again Corollary \ref{twirl2}, but with $Q$, $Q'$ being Pauli $Z$+identity operators and $\{P_i\}$ all tensor products of Pauli $X$+identity operators, and sum over $\boldsymbol{r'_{c(t)}}$ to eliminate the Pauli $Z$ components of the attack the differ in the two sides. Thus, given that $\sum_{u,k} |a_{u,k}|^2 = 1$ from the unital property of the attack, the attack becomes a convex combination of Pauli operators:
\begin{eqnarray}
F_c \geq \bigg( \sum_{\boldsymbol{t},\boldsymbol{b''_{c(t)}}} p(\boldsymbol{t}) \sum_{u,k} |a_{u,k}|^2    | \bigotimes_{i \in \boldsymbol{c(t)}} \bra{+}_i  E_G^{\dagger} \bigotimes_{i \in \boldsymbol{c(t)}} R_z(-\phi_i)  Z^{b''_i} \ket{+}\bra{+} Z^{b''_i} P_{k|i}  R_z(\phi_i) E_G  \bigotimes_{i \in \boldsymbol{c(t)}} \ket{+}_i |^2 \bigg)^{1/2} 
\end{eqnarray}

The Pauli $X$ component of $P_{k|i}$ can be replaced by $I$ since the only effect is, depending on $b''_i$, to change the sign of the quantity inside the absolute and the sign is eliminated. Then, we can sum over the $b$'s  to get identity and since $E_G^{\dagger} R_z(-\phi_i)$ now commutes with the attack Pauli operators, it cancels with $R_z(\phi_i) E_G$. Also, as before, we set $ |\alpha_{k}|^2 = \sum_{u} |a_{u,k}|^2 $
\begin{eqnarray}
F_c^2 \geq \sum_{\boldsymbol{t}} p(\boldsymbol{t}) \sum_{k} |\alpha_{k}|^2  \prod_{i \in \boldsymbol{c(t)}}  | \bra{+}_i     P_{k|i}    \ket{+}_i |^2. \label{eq2_} 
\end{eqnarray}

It is easy to see by the symmetry of the trap construction that, for attacks that are exactly the same on any of the $2\kappa+1$ graphs of the different the rounds of the protocol e.g. stochastic noise, Equation \ref{eq1_} and Equation \ref{eq2_} give the same result when averaged over $\boldsymbol{t}$. However, one needs to deal with more clever attacks which attack a different qubit at every round trying to coincide with dummies instead of traps with non-zero probability, i.e. for some of the possible permutations of the $2\kappa+1$ graphs.

We now show that for $\kappa \geq 1$ the maximum of $F_t^2-F_c^2$ for all possible deviation strategies  is
$$\Delta_{\kappa} \equiv \frac{\kappa! (\kappa+1)!}{(2\kappa+1)!}$$

The attack is a convex combination of Pauli operators, thus it suffices to find the Pauli operator that maximizes $F_t^2-F_c^2$. The  maximum comes from the attack that touches all $2 \kappa +1$ rounds. In this case, $F_c^2$ is lower bounded by $0$ and $F_t^2$ comes from the probability the attack does not coincide with any trap in the $2 \kappa$ trap computation rounds. There are $2 \kappa +1$ ways of picking the target computation round, which fixes the  positions of the $2 \kappa$ trap computation rounds. The choice of the even/odd parity positions for the traps is fixed not to coincide with the attacks. Further simplification can be done by observing that only the attacks on $\kappa+1$ of the same kind (even or odd) and $\kappa$ of the other, are successful. By attacking $\kappa+2$ or more of the same kind they are guaranteed to hit a trap of the same kind, independently of the position of the target. This reduces the possible ways of picking the target to $\kappa +1$, which gives a bound of $\frac{\kappa +1}{{2\kappa +1 \choose \kappa} (\kappa+1)}$, equal to the value of $\Delta_{\kappa}$ above.

To conclude the proof we show that for all other attacks  $F_t^2-F_c^2$ is non-positive and thus estimating $F^2_{c}$ through $F^2_{t}$ is a conservative estimation.  We begin by arguing that there is no benefit for the attacker to touch more than one qubit of each round, since the lower bound of the fidelity $F_c^2$ contains  products of terms that can be $0$ or $1$ and thus it suffices to make one term $0$ to make the product $0$. By symmetry of the construction it does not matter which particular qubit the attacker touches but only whether it is an odd or even position qubit at each round. Let $\lambda$ be the number of rounds that are attacked. 

Assume,  without loss of generality, the first $\xi$ rounds are attacked on an even qubit, where $\xi \leq \kappa$ otherwise it will certainly hit a trap, and $\lambda - \xi \leq \kappa$ for the same reason. Also, assume $\xi \geq (\lambda - \xi)$ without loss of generality.

For index $k$ to correspond to an attack on $\lambda$ rounds

\begin{eqnarray}
F_{c,k}^2 \geq 1 - \frac{\lambda}{2 \kappa +1} = \frac{2\kappa +1 - \lambda}{2 \kappa +1}
\end{eqnarray}

There are ${{2 \kappa +1}\choose{\kappa}} (\kappa+1)$ possibilities for the selection of traps. In order to count the combinations of attacks  not affecting the traps, we identify two cases, (i) the target is in the attacked rounds ($\lambda$ possible positions) and there are ${{2\kappa +1 - \lambda}\choose{\kappa -  \xi }}$ possible placings of the remaining traps in the non-attacked positions, (ii) the target is not in the attacked rounds ($2 \kappa +1 - \lambda$ possible positions) and there are ${{2\kappa - \lambda}\choose{\kappa - \xi}}$ possible placings of the remaining traps in the non-attacked positions:

\begin{eqnarray}
F_{t,k} = \frac{\lambda {{2\kappa +1 - \lambda}\choose{\kappa - \xi}} + (2 \kappa +1 - \lambda) {{2 \kappa - \lambda}\choose{\kappa - \xi}}} { {{2 \kappa +1}\choose{\kappa}} (\kappa+1)}
\end{eqnarray}

We show that for $\lambda \leq 2 \kappa$ we have $F_{t,k}-F_{c,k} \leq 0$.

\begin{eqnarray}
 \frac{\lambda {{2\kappa +1 - \lambda}\choose{\kappa - \xi}} + (2 \kappa +1 - \lambda) {{2 \kappa - \lambda}\choose{\kappa - \xi}}} { {{2 \kappa +1}\choose{\kappa}} (\kappa+1)} - \frac{2\kappa +1 - \lambda}{2 \kappa +1} \leq 0 \Leftrightarrow \nonumber \\
\frac{ (2 \kappa - \lambda +1 )(\kappa +\xi +1) {{2 \kappa  - \lambda}\choose{\kappa - \xi }}} {(\kappa - (\lambda - \xi) +1) {{2 \kappa +1}\choose{\kappa}} (\kappa+1)} - \frac{2\kappa +1 - \lambda}{2 \kappa +1} \leq 0 \Leftrightarrow \nonumber \\
\frac{ (2 \kappa - \lambda +1 )(\kappa + \xi +1)(2 \kappa - \lambda)! (\kappa !)^2 } {(\kappa - (\lambda - \xi) +1)  (2\kappa + 1)! (\kappa - (\lambda - \xi))!(\kappa - \xi)!} \leq \frac{2\kappa +1 - \lambda}{2 \kappa +1} \Leftrightarrow \nonumber \\
\frac{ (\kappa + \xi +1)(2 \kappa - \lambda)! (\kappa !)^2 } {(\kappa - (\lambda - \xi) +1)  (2\kappa)! (\kappa - (\lambda - \xi))!(\kappa - \xi)!} \leq 1
\end{eqnarray}

For $\kappa =\{1,2\}$ it is easy to verify the expression directly. For the general case  we rewrite the above as:

\begin{eqnarray}
\frac{ (\kappa + \xi +1) [(\kappa - \xi +1)\cdot \ldots \cdot \kappa ][(\kappa - (\lambda - \xi)+1) \cdot \ldots \cdot\kappa] } {(\kappa - (\lambda - \xi) +1)  (2\kappa - \lambda + 1) \cdot \ldots \cdot 2\kappa}  \leq 1
\end{eqnarray}

where we have $1+\xi+(\lambda-\xi)=1+ \lambda$ terms on the numerator and $1+\lambda$ terms in the denominator. 

For the LHS of the above equation we have

\begin{eqnarray}
\leq \frac{ (\kappa + \xi +1) } {(\kappa - (\lambda - \xi) +1)  2^{\xi}}  \nonumber \\
\leq \frac{ (\kappa + \xi +1) } {(\kappa - \xi +1)  2^{\xi}}
\end{eqnarray}

It suffices to show that the above is $\leq 1$ for all allowed values of $\kappa$ and $\xi$. We can rewrite it as:

\begin{eqnarray}
\frac{ (\kappa + \xi +1) } {(\kappa - \xi +1)  } \leq 2^{\xi} \Leftrightarrow \\
\frac{1}{\ln(2)}\ln(\frac{\kappa+1 + \xi}{\kappa+1 -\xi}) \leq \xi \Leftrightarrow \\
\frac{2}{\ln(2)} \text{artanh}\left(\frac{\xi}{\kappa+1}\right) \leq \xi
\end{eqnarray}

which is true for $\kappa \geq 3, \xi\leq \kappa$. This concludes our calculation of the bound $\Delta_{\kappa}$.

The rest of the proof is given in the main text.

\end{proof}

\section{\label{app:Thd}Calculation of FT Threshold for Protocol 2b}
\label{threshh}

Since  the logical graph is the brickwork state, its topological implementation will be on MBQC, as shown in Fig.~(\ref{3dgraph}), and will have the structure of Fig.~(\ref{bricktopo}).

Noise considered is local, unital and bounded. It applies after every elementary operation (preparation, entangling and measurement) $j$ and is expressed as a CPTP superoperator:

\begin{equation}
\mathcal{N}_j = (1-\epsilon) \mathcal{I} + \mathcal{E}_j
\label{noise_eq2}
\end{equation}

where $||\mathcal{E}_j||_{\diamond}=\epsilon$, where we set $\epsilon=\epsilon_V=\epsilon_P$ to calculate a common threshold for the verifier and the prover.

For the fault tolerant noisy, but honest, probability distribution post-selected for null syndrome measurement $q^{\text{nsy}}(\boldsymbol{x}|y=0)$,  and the exact one $ q^{\text{exc}}(\boldsymbol{x})$ we reproduce the derivation of Ref.~\cite{fujii2016noise}.

\begin{eqnarray}
\sum_{\boldsymbol{x}}  \left\lvert q^{\text{nsy}}(\boldsymbol{x}|y=0)-q^{\text{exc}}(\boldsymbol{x}) \right\rvert  &=& \sum_{\boldsymbol{x}} \left\lvert \frac{\text{Tr} (P_{\boldsymbol{x}} Q_y (\rho_{\text{sparse}} + \rho_{\text{faulty}}))}{\text{Tr}(Q_y(\rho_{\text{sparse}} + \rho_{\text{faulty}}))} - q^{\text{exc}}(\boldsymbol{x}) \right\rvert \nonumber \\
&=& \sum_{\boldsymbol{x}} \left\lvert \frac{\text{Tr} (P_{\boldsymbol{x}} Q_y  \rho_{\text{faulty}})}{\text{Tr}(Q_y(\rho_{\text{sparse}} + \rho_{\text{faulty}}))} - (1-b)q^{\text{exc}}(\boldsymbol{x})  \right\rvert  
\end{eqnarray}

Here, $P_{\boldsymbol{x}}$ is the projector to result $\boldsymbol{x}$ for the output register and $Q_y$ is the projector to null syndrome for the post-selection register. The  output quantum state of the sampler, just before the final measurements, can be written as a sum of two matrices:  $\rho_{\text{sparse}}$ that is the sum of the states on which apply the components of the noise operators $\mathcal{N}_j$ (i.e. either component $(1-\epsilon) \mathcal{I}$ or component $\mathcal{E}_j$ for each $j$) that, under postelection for $y=0$, do not produce a logical error in the output distribution, and $\rho_{\text{faulty}}$ which contains the sum of the states on which apply the rest of the noise components (for more detail see~\cite{fujii2016noise}). Thus, $\text{Tr}(P_{\boldsymbol{x}} Q_y \rho_{\text{sparse}}) \propto q^{\text{exc}}(\boldsymbol{x})$. Term $b$, which is defined by:
\begin{equation}
b =\frac{\text{Tr}(P_{\boldsymbol{x}} Q_y \rho_{\text{sparse}})}{\text{Tr}(Q_y (\rho_{\text{sparse}}+\rho_{\text{faulty}}))q^{\text{exc}}(\boldsymbol{x})}
\end{equation}

is therefore independent of $\boldsymbol{x}$.

By applying triangle inequality and by observing that the trace terms are positive

\begin{eqnarray}
\sum_{\boldsymbol{x}}  \left\lvert q^{\text{nsy}}(\boldsymbol{x}|y=0)-q^{\text{exc}}(\boldsymbol{x}) \right\rvert  \leq  \frac{\text{Tr} ( Q_y \rho_{\text{faulty}})}{\text{Tr}(Q_y(\rho_{\text{sparse}} + \rho_{\text{faulty}}))} + (1-b) \nonumber
\end{eqnarray}

Since $\sum_x \frac{\text{Tr} (P_x Q_y (\rho_{\text{sparse}} + \rho_{\text{faulty}}))}{\text{Tr}(Q_y(\rho_{\text{sparse}} + \rho_{\text{faulty}}))} =1$, we have $1-b = \frac{\text{Tr}(Q_y \rho_{\text{faulty}})}{\text{Tr}(Q_y(\rho_{\text{sparse}} + \rho_{\text{faulty}}))}$.

Also we have $\text{Tr}(Q_y(\rho_{\text{sparse}} + \rho_{\text{faulty}})) \geq (1-\epsilon)^N$. This comes from the fact there is at least one selection of components of the noise operators $\mathcal{N}_j$ that results in the null syndrome and this is components $(1-\epsilon) \mathcal{I}, \forall j$, giving a term with trace $(1-\epsilon)^N$. Thus

\begin{equation}
\sum_{\boldsymbol{x}}  \left\lvert q^{\text{nsy}}(\boldsymbol{x}|y=0)-q^{\text{exc}}(\boldsymbol{x}) \right\rvert   \leq 2 \text{Tr}( \rho_{\text{faulty}}) / (1-\epsilon)^N 
\end{equation}

In the case of the topological code the errors are created by error chains $\mathcal{L}$ of length greater than $L_d$, which is the minimum of the distance between two defects and the size of defects. Matrix $ \rho_{\text{faulty}}$ can be decomposed into terms $\rho_{\text{faulty}}^{\mathcal{L}}$ with respect to error chains $\mathcal{L}$ of length $L$ so that:

\begin{eqnarray}
\text{Tr} (\rho_{\text{faulty}} ) &\leq& \sum_{L=L_d}^{N} \sum_{\mathcal{L}:|\mathcal{L}|=L} \text{Tr} (\rho_{\text{faulty}}^{\mathcal{L}} )  \\
&\leq& \sum_{L=L_d}^{N} \sum_{\mathcal{L}:|\mathcal{L}|=L}  (1-\epsilon)^N \prod_{j=1}^{|\mathcal{L}|} \frac{|| \mathcal{E}_{j} ||_{\diamond}}{1-\epsilon}  \\
 &=& \sum_{L=L_d}^{N} \sum_{\mathcal{L}:|\mathcal{L}|=L}(1-\epsilon)^N \left(\frac{ \epsilon}{1- \epsilon} \right)^{|\mathcal{L}|} 
\end{eqnarray}

The number of error chains of length $|\mathcal{L}|$ in the 3D lattice of size $n$ is $\text{poly}(n) (6/5)5^{|\mathcal{L}|}$, which is the number of self avoiding walks~\cite{dennis2002}. Thus the gap $\sum_{\boldsymbol{x}}  \left\lvert q^{\text{nsy}}(\boldsymbol{x}|y=0)-q^{\text{exc}}(\boldsymbol{x}) \right\rvert $ is bounded by

\begin{equation}
 \leq 2   \sum_{L=d}^{N} \text{poly}(N) (6/5)5^{L} \left(\frac{\epsilon}{1- \epsilon} \right)^{L} 
\end{equation}

which converges to zero if $\epsilon / (1-\epsilon) < 1/5$. The threshold comes from the self-avoiding walks that affect the singular qubits and surpass the distillation threshold, where a more careful counting needs to be done~\cite{fujii2016noise} to get $\epsilon / (1-\epsilon) < 0.134$ or $\epsilon < 0.118$. This calculation of the threshold~\cite{fujii2016noise} is an underestimate because of the assumption that error correction is done on primal
and dual cubic lattices independently.

However, the above calculation is under the stochastic phenomenological noise model that does not account for the noise in the individual operations that compose a syndrome measurement. In our case, the topological code is implemented on the MBQC model where physical noise should be at least $6$ times less than phenomenological noise.
This is because there are typically $6$ operations involved in a syndrome measurement: $1$ syndrome qubit preparation, $4$ entangling operations with the surrounding qubits (less on boundaries) and $1$ syndrome measurement. This gives us threshold $\epsilon_{\text{thres}}=0.0196943$. 

The overhead of the error detection scheme comes by counting the number of error syndromes that are influencing one trap by catching potentially detectable errors (chains of size $\leq d$). This is the area of dimension $d$ around the defect qubits and their `past' in terms of MBQC flow (physical layer) which we choose in the smallest of the three dimensions of the topological code. In Fig.~(\ref{bricktopo})(c) we depict the logical qubits (made of prime/dual physical defects) that compose a `H' shaped component of the brickwork state. In our trappification scheme any topologically protected qubit can be a trap. For a worst case analysis we consider the biggest logical qubit, in terms of number of defects that make it, which is a prime qubit in our example (bottom middle blue loop in Fig.~(\ref{bricktopo})(c)).   

The number of  syndrome measurements (cubes) will depend on the distance parameter $d$. Since the counting of syndromes is involved we give an example for fixed values, $d=2$ and physical noise $\epsilon=(1/20)\epsilon_{\text{thres}}$. In this case the number of syndromes is at most $564$ and the number of repetitions is $M=1/(1-p_c)^{564}$, where $p_c$ is the probability of a cube syndrome failing. Probability $p_c$ is given by $(1-(1-2 (6\epsilon))^6)/2$, which is the probability of having an odd number of errors in the $6$ faces of a cube. This gives $M \approx 3 \times 10^8$. Overheads for other fractions of the threshold for noise appear in the main text.

\section{\label{app:Thm2}Proof of Theorem~\ref{thm2} soundness}

\begin{proof}

To establish soundness we need to show that a lower bound in the fidelity on the target computation round and the acceptance probability of the trap computation rounds are the same averaged over the random parameters.

The total variation distance between the experimental (noisy and potentially dishonest) distribution of the Ising sampler $q^{\mathrm{nsy}}(\boldsymbol{x}|y=0)$, where $y=0$ implies conditioning on the null syndrome, and the exact one $q^{\mathrm{exc}}(\boldsymbol{x}|y=0)=q^{\mathrm{exc}}(\boldsymbol{x})$  after the measurements is
\begin{eqnarray}
\text{var}^{\text{Post}} = \frac{1}{2} \sum_{\boldsymbol{x}} | q^{\mathrm{exc}}(\boldsymbol{x}|y=0) - q^{\mathrm{nsy}}(\boldsymbol{x}|y=0) |&=& \frac{1}{2} \sum_{\boldsymbol{x}} \left| q^{\mathrm{exc}}(\boldsymbol{x}) - \frac{q^{\mathrm{nsy}}(\boldsymbol{x},y=0)}{q^{\mathrm{nsy}}(y=0)} \right|   \\
&=&  D\left(\rho_c,\frac{\rho_c^{' \text{post}}}{q^{\mathrm{nsy}}(y=0)}\right)  \\
&\leq & \sqrt{1-F^2\left(\rho_c,\frac{\rho_c^{' \text{post}}}{q^{\mathrm{nsy}}(y=0)}\right)}  \\
&=&  \sqrt{1 - \text{Tr}^2 \left(\sqrt{ \frac{\rho_c \rho_c^{' \text{post}}}{q^{\mathrm{nsy}}(y=0)}} \right) },
\end{eqnarray} where $\rho_c$ is the correct state and $\rho^{' \text{post}}_c$ the experimental state, post-selected on the null syndrome measurements, after all measurements. For the rest of this section we denote $q^{\mathrm{nsy}}(y=0)$ as $q'_0$ for simplicity. 

For the target round, the average fidelity $F_c$ is calculated in the physical level of the computation as in the non-fault-tolerant case.  The qubits are pre-rotated by $\theta_i$, or flipped by $d_i$ in the case of dummies.

Noise can enter either during the state preparation from the verifier, or during the single round elementary MBQC operations (entangling and measurement) of the prover. We assume a noise model which is local, unital and bounded, so that standard fault tolerance techniques are applicable. Noise can always moved after every elementary operation on qubit $j$ and expressed as a CPTP superoperator applies only on the state of qubit $j$:

\begin{equation}
\mathcal{N}_j = (1-\epsilon) \mathcal{I} + \mathcal{E}_j
\end{equation}

where $||\mathcal{E}_j||_{\diamond}\leq \epsilon_{\text{thres}}$.

Crucially, we assume that the noise during the preparation does not have any dependence on the secret parameter $\theta_i$. 

Moving all the noise operators just before the measurement, results to a different set of local, unital and bounded operators $\mathcal{N'}_j$, collectively represented as $\mathcal{N}'$.

We apply the same twirling steps as in the proof of Theorem~\ref{thm1} to twirl the CPTP map that is the composition of the attack and the noise. Notice that the twirl on the post-selected qubits is trivial since there is no sum over $b'_i$. Thus,
\begin{eqnarray}
F_c^2 \left( \rho_c^{\text{post}},\frac{\rho_c^{' \text{post}}}{q'_0}\right) & \geq & \frac{1}{q'_0} \sum_{\boldsymbol{t},\boldsymbol{b'_{c(t)}}} p(\boldsymbol{t}) \sum_{u,k} |a_{u,k}|^2    \left| \bigotimes_{i \in \boldsymbol{c(t)}} \bra{0}_i \bigotimes_{i \in \boldsymbol{c(t)}} \bra{+}_i  E_G^{\dagger} \bigotimes_{i \in \boldsymbol{c(t)}} R_z(-\phi_i)  Z^{b'_i} \ket{+}_i
\bra{+}_i Z^{b'_i}     \right. \nonumber \\
&&  \left.  P_{k|i} R_z(\phi_i) E_G \bigotimes_{i \in \boldsymbol{c(t)}} \ket{+}_i \bigotimes_{i \in \boldsymbol{c(t)}} \ket{0}_i  \right|^2, 
\end{eqnarray}
 where $b'_i$'s take fixed values in the sum for the syndrome measurements such that the syndrome indicates null errors.

The only (noise and attack) Pauli operators that have an effect on the above quantity are tensor products of identity and Pauli $Z$. These operators flip the measurement outcome of the particular qubit. Detectable attacks disappear because of the projector to null syndromes. The undetected attacks that come from operators $ P_{k|i}$ can be written as logical bit flips on the subsequent measurements - since it will affect the classical post-processing.
 Also, the normalization factor vanishes when we trace over the syndrome systems.

Therefore, at the logical level  
\begin{eqnarray}
F_c^2 &\geq& \sum_{\boldsymbol{t},\boldsymbol{b''_{c(t)}}} p(\boldsymbol{t}) \sum_{u,k} |a_{u,k}|^2    \left| \bigotimes_{i \in \boldsymbol{c(t)} } \bra{0}_{i}^L \bigotimes_{i \in \boldsymbol{c(t)}} \bra{+}_{i}^L  E_{G}^{\dagger L} \bigotimes_{i \in \boldsymbol{c(t)}} R_{z}(-\phi_{i}^L)  Z^{b''_i L}  \ket{+}^{L}\bra{+}^{L} Z^{b''_i L}\right. \nonumber  \\
&& \left.    P_{k|i}^L  R_{z}(\phi_i^L) E_{G}^L \bigotimes_{i \in \boldsymbol{c(t)}} \ket{+}_{i}^{L}  \bigotimes_{i \in \boldsymbol{c(t)}} \ket{0}_{i}^{L} \right|^2, 
\end{eqnarray}
 We can now sum over the index $b''_{c(t)}$ to simplify the expression, by cancelling also the rotation and entangling operators. On the logical dummy system the logical Pauli $Z$ attacks have no effect, therefore it has trace $1$ and can be simplified to
\begin{equation}
F_c^2 \geq \sum_{\boldsymbol{t}} p(\boldsymbol{t}) \sum_{u,k} |a_{u,k}|^2  \prod_{i \in \boldsymbol{c(t)}}  | \bra{+}_{i}^L     P_{k|i}^{L}    \ket{+}_{i}^{L} |^2. \label{eq2} 
\end{equation}
The same technique can be employed for the trap rounds, with the difference that instead of  post-selection there is error correction for Protocol 2a and error detection for Protocol 2b that results in the same logical state for the same (noise and attack) Pauli operators.

From completeness we have set the limit of acceptance of the fidelity estimate to $(1- 2 \epsilon'')$. By repeating $N = \log(1/\beta)/(2 \epsilon''^2)$ times gets us $\sqrt{\epsilon''}$-close in our estimation with confidence $1 - \beta$. Thus, with this confidence we get  bound $\text{var}^{\text{Post}} \leq \sqrt{3 \epsilon'' + \Delta_{\kappa}}$.

\end{proof}

\section{\label{app:lemma}Channel Twirl Lemma}

The following lemma is used in the verifiability proofs.

\begin{lemma}
\begin{equation}
\sum_{i=1}^{4^n} P_i Q P_i \rho P_i Q' P_i = 0, \text{ if } Q \neq Q' \label{use_prop}
\end{equation}
where $\rho$ is a matrix of dimension $2^n \times 2^n$, $Q$, $Q'$ are two arbitrary $n$-fold tensor products of Pauli+identity operators $\{I,X,Y,Z\}$, and $\{P_i\}$ is the set of all $n$-fold tensor products of Pauli operators and the identity $\{I,X,Y,Z\}$.
\label{twirl1}
\end{lemma}

A proof of this lemma is also provided in Ref.~\cite{dankert09operator_twirl}.

\begin{proof}

We can write $Q$ as $Z_{\boldsymbol{a}} X_{\boldsymbol{a}'} = Z^{a_1} \otimes \ldots \otimes Z^{a_n} X^{a'_1} \otimes \ldots \otimes X^{a'_n} $, for arbitrary binary vectors $\boldsymbol{a} = (a_1, \ldots, a_n)$, $\boldsymbol{a}' = (a'_1, \ldots, a'_n)$, and similarly $Q'=Z_{\boldsymbol{b}}X_{\boldsymbol{b}'}$. 
Assuming $Q \neq Q',$ either $\boldsymbol{a} \neq \boldsymbol{b}$ or $\boldsymbol{a}' \neq \boldsymbol{b}'$. Summing over all $P_{\boldsymbol{k},\boldsymbol{k}'}$'s which are the $n$-fold tensor products of the form $Z_{\boldsymbol{k}} X_{\boldsymbol{k}'}= Z^{k'_1} \otimes \ldots \otimes Z^{k'_n}X^{k_1} \otimes \ldots \otimes X^{k_n}$ for binary vectors $\boldsymbol{k} = (k_1, \ldots, k_n), \boldsymbol{k}' = (k'_1, \ldots, k'_n),$ we get

\begin{eqnarray}
&&\sum_{\boldsymbol{k}, \boldsymbol{k}'} P_{\boldsymbol{k}, \boldsymbol{k}'} Q P_{\boldsymbol{k}, \boldsymbol{k}'} \rho P_{\boldsymbol{k}, \boldsymbol{k}'} Q' P_{\boldsymbol{k}, \boldsymbol{k}'}  \nonumber \\
&&= \sum_{\boldsymbol{k}, \boldsymbol{k}'} Z_{\boldsymbol{k}} X_{\boldsymbol{k}'} Z_{\boldsymbol{a}} X_{\boldsymbol{a}'} Z_{\boldsymbol{k}} X_{\boldsymbol{k}'}  \rho 
  Z_{\boldsymbol{k}} X_{\boldsymbol{k}'} Z_{\boldsymbol{b}} X_{\boldsymbol{b}'} Z_{\boldsymbol{k}} X_{\boldsymbol{k}'}  \\
&&  = \sum_{\boldsymbol{k}, \boldsymbol{k}'} Z_{\boldsymbol{k}} (X_{\boldsymbol{k}'} Z_{\boldsymbol{a}} X_{\boldsymbol{k}'} ) X_{\boldsymbol{a}'} Z_{\boldsymbol{k}}   \rho 
  Z_{\boldsymbol{k}} (X_{\boldsymbol{k}'} Z_{\boldsymbol{b}} X_{\boldsymbol{k}'}) X_{\boldsymbol{b}'} Z_{\boldsymbol{k}}    \\
&& = \sum_{\boldsymbol{k},\boldsymbol{k}'} Z_{\boldsymbol{k}} ((-1)^{\boldsymbol{k}' \cdot \boldsymbol{a}} Z_{\boldsymbol{a}}) X_{\boldsymbol{a}'} Z_{\boldsymbol{k}} \rho Z_{\boldsymbol{k}}((-1)^{\boldsymbol{k}' \cdot \boldsymbol{b}} Z_{\boldsymbol{b}})X_{\boldsymbol{b}'} Z_{\boldsymbol{k}}  \\
&& = \sum_{\boldsymbol{k},\boldsymbol{k}'} (-1)^{\boldsymbol{k}' \cdot (\boldsymbol{a} \oplus \boldsymbol{b})}  Z_{\boldsymbol{a}} (Z_{\boldsymbol{k}} X_{\boldsymbol{a}'} Z_{\boldsymbol{k}}) \rho  Z_{\boldsymbol{b}} (Z_{\boldsymbol{k}} X_{\boldsymbol{b}'} Z_{\boldsymbol{k}})  \\
&& = \sum_{\boldsymbol{k}'} (-1)^{\boldsymbol{k}' \cdot (\boldsymbol{a} \oplus \boldsymbol{b})} \sum_{\boldsymbol{k}} (-1)^{\boldsymbol{k} \cdot (\boldsymbol{a}' \oplus \boldsymbol{b}')}  Z_{\boldsymbol{a}}  X_{\boldsymbol{a}'}  \rho  Z_{\boldsymbol{b}}  X_{\boldsymbol{b}'}  
\end{eqnarray} 

If either $\boldsymbol{a} \neq \boldsymbol{b}$ or $\boldsymbol{a}' \neq \boldsymbol{b}'$ the summation $\sum_{\boldsymbol{k}'} ((-1)^{\boldsymbol{k}' \cdot (\boldsymbol{a} \oplus \boldsymbol{b})})$ or $\sum_{\boldsymbol{k}} ((-1)^{\boldsymbol{k} \cdot (\boldsymbol{a}' \oplus \boldsymbol{b}')})$ is equal to zero respectively, because (in either case) exactly half of the elements of the summation will be $-1$ and half will be $1$. Therefore, since our assumption was that either $\boldsymbol{a} \neq \boldsymbol{b}$ or $\boldsymbol{a}' \neq \boldsymbol{b}'$ or both, the whole expression equals zero.

\end{proof}

\begin{corollary} \label{twirl2}
\begin{equation}
\sum_{i=1}^{2^n} P_i Q P_i \rho P_i Q' P_i = 0, \text{ if } Q \neq Q' \label{use_prop}
\end{equation}
where $\rho$ is a matrix of dimension $2^n \times 2^n$, $Q$, $Q'$ are two arbitrary $n$-fold tensor products of Pauli $X$ and identity operators $\{I,X\}$, and $\{P_i\}$ is the set of all $n$-fold tensor products of Pauli $Z$ and identity operators $\{I,Z\}.$
\end{corollary}

\begin{proof}

Since $Q \neq Q',$ Lemma \ref{twirl1} gives

\begin{equation}
\sum_{i=1}^{4^n} P_i Q P_i \rho P_i Q' P_i = 0 
\label{use_prop}
\end{equation}

where $\{P_i\}$ is the set of all $n$-fold tensor products of the Pauli operators and the identity  $\{I,X,Y,Z\}$. But since $Q$ and $Q'$ have only identity and Pauli $X$ tensor elements the Pauli $X$ operators of $\{P_i\}$ commute with $Q$ and $Q'$ on each side and give identity. %
\end{proof}

\end{document}